\documentclass[preprint,11pt]{elsarticle}

\usepackage[margin=0.95in]{geometry}
\usepackage{multicol,vwcol,soul}
\usepackage{latexsym,amsmath,amssymb,stmaryrd,mathtools}
\usepackage{graphicx,epsfig,tikz}
\usepackage{makecell, longtable, lscape, ctable,lscape, supertabular}
\usepackage{subfig}
\usepackage{csvsimple}
\usetikzlibrary{shapes.geometric,positioning,chains,fit,calc}
\tikzset{near start abs/.style={xshift=1cm},
	node/.style={circle,draw},
	nodeone/.style={circle,draw,gray, line width=0.7mm},
	nodered/.style={circle,draw,red, line width=0.7mm}}
\usepackage[ddmmyyyy]{datetime}
\usepackage[ruled,noend,noline]{algorithm2e}
\usepackage{enumitem}
\usepackage{cleveref}

\makeatletter
\newdimen\commentwd
\let\oldtcp\tcp
\def\tcp*[#1]#2{
	\setbox0\hbox{#2}%
	\ifdim\wd\z@>\commentwd\global\commentwd\wd\z@\fi
	\oldtcp*[r]{\leavevmode\hbox to \commentwd{\box0\hfill}}}

\let\oldalgorithm\algorithm
\def\algorithm{\oldalgorithm
	\global\commentwd\z@
	\expandafter\ifx\csname commentwd@\romannumeral\csname c@\algocf@float\endcsname\endcsname\relax\else
	\global\commentwd\csname commentwd@\romannumeral\csname c@\algocf@float\endcsname\endcsname
	\fi
}
\let\oldendalgorithm\endalgorithm
\def\endalgorithm{\oldendalgorithm
	\immediate\write\@auxout{\gdef\expandafter\string\csname commentwd@\romannumeral\csname c@\algocf@float\endcsname\endcsname{%
			\the\commentwd}}}

\usepackage{comment}

\graphicspath{ {./Images/} }

\renewcommand{\epsilon}{\varepsilon}
\renewcommand{\phi}{\varphi}
\newcommand{\BWT}{\textrm{BWT}}
\newcommand{\dol}{\textit{dol}}

\newcommand{\shift}{\textit{shift}}
\newcommand{\unshift}{\textit{unshift}}

\newcommand{\lex}{\textrm{lex}}
\newcommand{\SigmaD}{\Sigma^*_{\$}}
\newcommand{\sgn}{\textit{sgn}}
\newcommand{\Oh}{\mathcal{O}}

\newenvironment{proof}{{\em Proof.}}{\hspace*{\fill}$\Box$\par\vspace{3mm}}

\let\oldnl\nl
\newcommand{\nonl}{\renewcommand{\nl}{\let\nl\oldnl}}

\usepackage{scalerel}

\newcommand\fat[1]{\ThisStyle{\hstretch{1.2}{\ooalign{%
				\kern.46pt$\SavedStyle#1$\cr\kern.33pt$\SavedStyle#1$\cr%
				\kern.2pt$\SavedStyle#1$\cr$\SavedStyle#1$}}}}

\usepackage{newfloat}
\DeclareFloatingEnvironment[
name=Table]{newtable}

\def\a{{\tt a}}
\def\b{{\tt b}}

\def\x{{\tt x}}

\def\n{{\tt n}}

\def\links{{\textrm{left}}}
\def\rechts{{\textrm{right}}}
\def\sgn{{\textrm{sgn}}}

\newtheorem{theorem}{Theorem}
\newtheorem{proposition}{Proposition}
\newtheorem{lemma}{Lemma}

\newtheorem{corollary}{Corollary}

\newtheorem{definition}{Definition}
\newtheorem{exa}{Example}

\usepackage{xcolor}
\newcommand{\red}{\textcolor{red}}

\newcommand{\powerFormbis}{(1, e_1, \ldots, e_m)(2, e_1+1, \ldots, e_m+1)\ldots(c, e_1+c-1, \ldots, e_m+c-1)}

\makeatletter
\def\ps@pprintTitle{%
	\let\@oddhead\@empty
	\let\@evenhead\@empty
	\def\@oddfoot{\footnotesize\itshape
		{Published in Theoretical Computer Science} \hfill\today}%
	\let\@evenfoot\@oddfoot
}
\makeatother

\begin{document}
	
	\begin{frontmatter}
	\title{When a Dollar Makes a BWT}
	\author{Sara Giuliani}
	\ead{sara.giuliani\_01@univr.it}
	\author{Zsuzsanna Lipt\'ak\corref{cor1}}
	\ead{zsuzsanna.liptak@univr.it}
	\author{Francesco Masillo}
	\ead{francesco.masillo@studenti.univr.it}
	\author{Romeo Rizzi}
	\ead{romeo.rizzi@univr.it}
	
	\cortext[cor1]{Corresponding author}
	
	\address{Department of Computer Science, University of Verona, Italy}

\begin{abstract}
	
	The Burrows-Wheeler-Transform (BWT) is a reversible string transformation which plays a central role in text compression and is fundamental in many modern bioinformatics applications. The BWT is a permutation of the characters, which is in general better compressible and allows to answer several different query types more efficiently than the original string. 
	
	It is easy to see that not every string is a BWT image, and exact characterizations of BWT images are known. We investigate a related combinatorial question. In many applications, a sentinel character \$ is added to mark the end of the string, and thus the BWT of a string ending with $\$$ contains exactly one \$-character. Given a string $w$, we ask in which positions, if any, the \$-character can be inserted to turn $w$ into the BWT image of a word ending with $\$$. We show that this depends only on the standard permutation of $w$ and present a $\mathcal{O}(n \log n)$-time algorithm for identifying all such positions, improving on the naive quadratic time algorithm. We also give a combinatorial characterization of such positions and develop bounds on their number and value. This is an extended version of [Giuliani et al.\ ICTCS 2019].
\end{abstract}


\end{frontmatter}

\section{Introduction}\label{sec:introduction}

The Burrows-Wheeler-Transform (BWT), introduced by Burrows and Wheeler in 1994~\cite{BurrowsWheeler94}, is a reversible string transformation which is fundamental in string compression and is at the core of many of the most frequently used bioinformatics tools~\cite{bwa,bowtie,soap2}. The BWT, a permutation of the characters of the original string, is particularly well compressible if the original string has many repeated substrings, thus making it highly relevant for natural language texts and for biological sequence data. This is due to what is sometimes referred to as {\em clustering effect}~\cite{RosoneS13}: repeated substrings cause equal characters to be grouped together, resulting in longer runs of the same character than in the original string, and as a result, in higher compressibility.

\medskip

Given a word (or string) %
$v$ over a finite ordered alphabet, the BWT is a permutation of the characters of $v$,  such that position $i$ contains the last character of the $i$th ranked rotation of $v$, with respect to lexicographic order, among all rotations of $v$. For example, the BWT of the word ${\tt banana}$ is ${\tt nnbaaa}$, see Fig.~\ref{fig:bwt-ex} (left).  A fundamental property of the BWT is that it is {\em reversible}: Given a BWT image $w$, a word $v$ such that $\BWT(v)=w$ can be found in linear time in the length of $w$, and $v$ is unique up to rotation~\cite{BurrowsWheeler94}. 

\begin{figure}
	\centering
	\begin{minipage}[t]{4cm}
		\begin{tabular}{c@{\hspace{.5cm}}c}
			rotations &\\
			of {\tt banana} & \BWT\\
			\hline
			{\tt abanan} & \n\\
			{\tt anaban} & \n\\
			{\tt ananab} & \b\\
			{\tt banana} & \a\\
			{\tt nabana} & \a\\
			{\tt nanaba} & \a
		\end{tabular}
	\end{minipage}
	\begin{minipage}[t]{4cm}
		\begin{tabular}{c@{\hspace{.5cm}}c}
			rotations & \\
			of ${\tt nanana}$ & \BWT\\
			\hline
			{\tt ananan} & \n\\
			{\tt ananan} & \n\\
			{\tt ananan} & \n\\
			{\tt nanana} & \a\\
			{\tt nanana} & \a\\
			{\tt nanana} & \a
		\end{tabular}
	\end{minipage}
	\begin{minipage}[t]{4cm}
		\begin{tabular}{c@{\hspace{.5cm}}c}
			rotations & \\
			of ${\tt nanana\$}$ & \BWT\\
			\hline
			{\tt \$nanana} & \a\\
			{\tt a\$nanan} & \n\\
			{\tt ana\$nan} & \n\\
			{\tt anana\$n} & \n\\
			{\tt na\$nana} & \a\\
			{\tt nana\$na} & \a\\
			{\tt nanana\$} & \$
		\end{tabular}
	\end{minipage}
	\caption{BWT of the strings {\tt banana, nanana} and {\tt nanana\$}. \label{fig:bwt-ex}}
\end{figure}

The BWT is defined for every word, even if not all of its rotations are distinct; this is the case, for example, with the word {\tt nanana}, whose  BWT is {\tt nnnaaa}, see Fig.~\ref{fig:bwt-ex} (center). 
(Words for which all rotations are distinct are called {\em primitive}.) %
On the other hand, not every word is a  BWT image, i.e.\ not every word is the BWT of some word. 
For example, {\tt banana} is not the BWT of any word. 

It can be decided algorithmically whether a given word $w$ is a BWT image, by slightly modifying the above-mentioned reversal algorithm: if $w$ is not a BWT image, then the algorithm terminates in $\Oh(n)$ time with an error message, where $n$ is the length of $w$.  Combinatorial characterizations of BWT images are also known~\cite{MantaciRS03,LShur11}: whether $w$ is a BWT image, depends on the number and characteristics of the cycles of its {\em standard permutation} (see Sec.~\ref{sec:basics}). In particular, $w$ is the BWT image of a primitive word if and only if its standard permutation is cyclic. Moreover, a necessary condition is that the runlengths of $w$ be co-prime~\cite{LShur11}. 

Another fundamental result in this context concerns the inverse of the standard permutation: in~\cite{CrochemoreDP05} it was shown that a permutation $\pi$ is the inverse of the standard permutation of the BWT of a primitive word if and only if $\pi$ is cyclic and it has at most $|\Sigma|-1$ descents. This result is a special case of a more general theorem by Gessel and Reutenauer~\cite{GesselR93}. 

In many situations, it is convenient to append a sentinel character $\$$ to mark the end of the word $v$; this sentinel character is defined to be lexicographically smaller than all characters from the given alphabet. For example, $\BWT({\tt nanana\$}) = {\tt annnaa\$}$,  see Fig.~\ref{fig:bwt-ex} (right). Adding the sentinel character, among other things, has the effect that the lexicographic order of the rotations of the word equals the lexicographic order of its suffixes, in this way allowing the use of suffix sorting in BWT construction. 
Clearly, all rotations of $v\$$ are distinct, thus the inverse of the BWT becomes unique, due to the condition that the sentinel character must be at the end of the word. In other words, given a word $w$ with exactly one occurrence of $\$$, there exists at most one word $v$ such that $w = \BWT(v\$)$. 

In this paper, we ask the following combinatorial question: Given a word $w$ over alphabet $\Sigma$, in which positions, if any, can we insert the $\$$-character such that the resulting word is the BWT image of some word $v\$$? We call such positions {\em nice.} Returning to our earlier examples: there are two nice positions for the word {\tt annnaa}, namely $3$ and $7$: {\tt an\$nnaa} and {\tt annnaa\$} are BWT images. However, there is none for the word {\tt banana}: in no position can $\$$ be inserted  such that the resulting word becomes a BWT image. 

We are interested both in characterizing nice positions for a given word $w$, and in computing them. Note that using the BWT reversal algorithm, these positions can be computed naively in $\Oh(n^2)$ time. Our results are as follows: 

\begin{itemize}
	\item we show that the question which positions are nice depends only on the standard permutation of $w$; 
	\item we present an $\Oh(n \log n)$ time algorithm to compute all nice positions of an $n$-length word $w$; and 
	\item we give a full combinatorial characterization of nice positions, via certain subsets which form what we call {\em pseudo-cycles} of the standard permutation of $w$. 
	
\end{itemize}

\subsection{Related work}

The BWT has been subject of intense research in the last two decades, from compression~\cite{Manzini01,FerraginaGMS05,KaplanV07,KaplanLV07}, algorithmic~\cite{CrochemoreGKL15,LouzaGT17,PolicritiP18}, and combinatorial~\cite{GiancarloRS07,RestivoR11,DaykinGGLLLP18} points of view (mentioning just a tiny selection from the recent literature). It has also been extended in several ways. One of these, the extended BWT, generalizes the BWT to a multiset of strings~\cite{MantaciRRS07,MantaciRRS08,BonomoMRRS14}, with successful applications to several  bioinformatics problems~\cite{MantaciRRS08,CoxJRS12,PrezzaPSR19}. A very recent development is the introduction of Wheeler graphs~\cite{GagieMS17}, a generalization of a fundamental underlying property of the BWT to data other than strings.%

There has been much recent work on inferring strings from different data structures built on strings (sometimes called reverse engineering), and/or just deciding whether such a string exists, given the data structure itself. For instance, this question has been studied for directed acyclic word graphs (DAWGs) and suffix arrays~\cite{BannaiIST03}, prefix tables~\cite{ClementCR09}, LCP-arrays~\cite{KarkkainenPP17}, Lyndon arrays~\cite{DaykinFHIS18}, and suffix trees~\cite{IIBT14,StarikovskayaV15,CazauxR14}. A number of papers study which permutations are suffix arrays of some string~\cite{HeMR05,SchurmannS08,KucherovTV13}, giving a full characterization in terms of the standard permutation. 

The analogous question for BWT images was answered fully in~\cite{MantaciRS03} for strings over binary alphabets, and in~\cite{LShur11} for strings over general alphabets. In~\cite{LShur11}, the authors also asked the question which strings can be ''blown up'' to become a BWT: Given the runs (blocks of equal characters) in $w$, when does a BWT image exist whose runs follow the same order, but each run can be of the same length or longer than the corresponding one in $w$? The authors fully characterize such strings, showing that the non-existence of a global ascent in $w$ is a necessary and sufficient condition. %

Another work treating a related question to ours is~\cite{MantaciRRRS17}, where the authors ask and partially answer the question of which strings are fixpoints of the BWT. Finally, a topic that has received considerable interest are words whose BWT has as few runs as possible: in~\cite{SimpsonP08,RestivoR09}, words with so-called {\em simple} BWTs were studied, i.e.\ words whose BWT has the form $a_k^{n_k}a_{k-1}^{n_{k-1}}\cdots a_1^{n_1}$, over an alphabet $\Sigma = \{a_1<a_2<\ldots < a_k\}$. The authors of~\cite{FerencziZ13} give a full characterization of words with {\em fully clustering} BWTs, i.e.\ of words whose BWT has the same number of runs as there are distinct characters, but not necessarily in decreasing order.

\subsubsection*{Overview}
The paper is organized as follows. In Section~\ref{sec:basics} we provide the necessary background and terminology. Next we present our algorithm for computing all nice positions of a string $w$ (Section~\ref{sec:algo}). In Section~\ref{sec:characterization}, we give a complete characterization of nice positions, followed in Section~\ref{sec:parity} by some bounds on the number and value of nice positions of a word. In Section~\ref{sec:results} we give some experimental results. 
We close with a discussion and outlook in Section~\ref{sec:conclusion}. Additional examples, extensive statistics, and further technical details are contained in the Appendix. 

%
%

\section{Basics}\label{sec:basics}

In this section we give the necessary terminology and notation. 

\subsection{Words} %
Let $\Sigma$ be a finite ordered alphabet. A {\em word} (or {\em string}) %
over $\Sigma$ is a finite sequence of elements from $\Sigma$ (also called {\em characters}). We write words as $w=w_1\cdots w_n$, with $w_i$ the $i$th character, and $|w|=n$ its {\em length}. Note that we index words from $1$. The {\em empty string} is the only string of length $0$ and is denoted $\epsilon$. The set of all words over $\Sigma$ is denoted $\Sigma^*$. The concatenation $w=uv$ of two words $u,v$ is defined by $w = u_1\cdots u_{|u|}v_1\cdots v_{|v|}$. Let $w=uxv$, with $u,x,v$ possibly empty. Then $u$ is called a {\em prefix}, $x$ a {\em factor} (or {\em substring}), and $v$ a {\em suffix} of $w$. A factor (prefix, suffix) $u$ of $w$ is called {\em proper} if $u\neq w$. For a word $u$ and an integer $k\geq 1$, $u^k = u\cdots u$ denotes the $k$-fold concatenation of $u$. A word $w$ is called a {\em primitive} if $w=u^k$ implies $k=1$. A {\em run} in a word $w$ is a maximal substring of the form $a^k$ for some $a\in\Sigma$, and a {\em runlength} is the length of such a maximal substring.

Two words $w, w'$ are called {\em conjugates} if there exist words $u,v$, possibly empty, such that $w=uv$ and $w'=vu$. Conjugacy is an equivalence relation, and the set of all words which are conjugates of $w$ constitute $w$'s {\em conjugacy class}. Given a word $w=w_1\cdots w_n$, the {\em $i$'th rotation} of $w$ is $w_i\cdots w_nw_1\cdots w_{i-1}$. Clearly, two words are conjugates if and only if one is a rotation of the other. 

The set of all words over $\Sigma$ is totally ordered by the {\em lexicographic order:} Let $v,w \in \Sigma^*$, then $v \leq_{\lex} w$ if $v$ is a prefix of $w$, or there exists an index $j$ s.t.\ for all $i<j$, $v_i = w_i$, and $v_j < w_j$ according to the order on $\Sigma$. A word $w$ of length $n$ is called {\em Lyndon} if it is lexicographically strictly smaller than all of its conjugates $v\neq w$. 

In the context of string data structures, it is often necessary to mark the end of words in a special way. To this end, let $\$ \not\in \Sigma$ be a new character, called {\em sentinel}, and set $\$ < a$ for all $a\in \Sigma$. Let $\SigmaD$ denote the set of all words over $\Sigma$ with an additional $\$$ at the end. The mapping $w\mapsto w\$$ is a bijection from $\Sigma^*$ to $\SigmaD$. Clearly, every word in $\SigmaD$ is primitive.

\subsection{Permutations} 
Let $n$ be a positive integer. A {\em permutation} is %
a bijection from $\{1,2,\ldots,n\}$ to itself. 
Permutations are often written using the two-line notation $\bigl(\begin{smallmatrix} 1 & 2 & \ldots & n \\ \pi(1) & \pi(2) & \ldots & \pi(n) \end{smallmatrix}\bigr) $. A {\em cycle} in a permutation $\pi$ is a minimal subset $C \subseteq \{1,\ldots,n\}$ with the property that $\pi(C)=C$. A cycle of length $1$ is called a {\em fixpoint}, and one of length $2$ a {\em transposition}.  Every permutation can be decomposed uniquely into disjoint cycles, giving rise to the {\em cycle representation} of a permutation $\pi$, i.e.\ as a composition of the cycles in the cycle decomposition of $\pi$. %
For example, $\pi = \bigl( \begin{smallmatrix} 1 & 2 & 3 & 4 & 5 & 6 \\ 4 & 2 & 5 & 6 & 3 & 1 \end{smallmatrix}\bigr) = (1 \: 4 \: 6)(2)(3 \: 5)$. Permutations whose cycle decomposition consists of just one cycle are called {\em cyclic}.

A fundamental theorem about permutations says that every permutation $\pi$ can be written as a product (composition) of transpositions, and that the number of any sequence of transpositions whose product is $\pi$ is either always even or always odd: this is called the {\em parity} of the permutation. The {\em sign} $\sgn(\pi)$ of a permutation $\pi$ is defined as $1$ if $\pi$ is even, and as $(-1)$ if it is odd; equivalently, $\sgn(\pi) = (-1)^m$, where $\pi = \prod_{i=1}^m \tau_i$ for some transpositions $\tau_i$. 

The sign of a cycle of $m$ elements is $(-1)^{m-1}$, since any cycle $C= (x_1, \ldots, x_m)$ can be written as $C  = (x_1, x_2)(x_2,x_3)\cdots (x_{m-1},x_m)$, i.e.\ $C$ is the product of $m-1$ transpositions. Moreover, if $\pi = \prod_{i=1}^c C_i$ is the cycle decomposition of permutation $\pi$ of $\{1,\ldots, n\}$, then $\sgn(\pi) = \prod_{i=1}^c \sgn(C_i) = (-1)^{n - c}$. 
For more details on permutations, see~\cite{Bona12}.

For a finite set of positive integers $X$ we denote by $\min X=\min\{x \mid x\in X\}$, $\max X=\max\{x \mid x\in X\}$, and for $y \in \mathbb{Z}$, $X+y =\{x+y \mid x \in X\}$. 

Finally, given a word $w$, the {\em standard permutation} of $w$, denoted $\sigma_w$, is the permutation defined by: $\sigma_w(i) < \sigma_w(j)$ if and only if either $w_i < w_j$, or $w_i = w_j$ and $i<j$. For example, the standard permutation of {\tt banana} is $\bigl( \begin{smallmatrix} 1 & 2 & 3 & 4 & 5 & 6 \\ 4 & 1 & 5 & 2 & 6 & 3 \end{smallmatrix}\bigr)$.

\subsection{Burrows-Wheeler-Transform} 
It is easiest to define the Burrows-Wheeler-Transform (BWT)~\cite{BurrowsWheeler94} via a construction: Let $v\in \Sigma^*$ with $|v|=n>0$, and let $M$ be an $n \times n$-matrix containing as rows all $n$ rotations of $v$ (not necessarily distinct) in lexicographic order (see Fig.~\ref{fig:bwt-ex}). Then $w=\BWT(v)$ is the last column of $M$. If $v$ is primitive, then this is equivalent to saying that $w=w_1\cdots w_n$ such that $w_i$ equals the last character of the $j$th rotation of $v$, where the $j$th rotation has rank $i$ among all rotations of $v$ w.r.t.\ lexicographic order. %

Linear-time construction algorithms of the BWT are well-known~\cite{RosoneS13}, and the BWT is {\em reversible}: Given a word $w$ which is the BWT of some word $v$, $v$ can be recovered from $w=\BWT(v)$, uniquely up to its conjugacy class, again in linear time. 

We briefly recap the algorithm for the case where $w$ is the BWT of a primitive word $v$, since this is the case we will need in the following.  The algorithm is based on the following insights about the matrix $M$: (1) the last character in each row is the one {\em preceding} the first character in the same row, (2) since the rows are rotations of the same word, every character in the last column occurs also in the first column, (3) the first column lists the characters of $v$ in lexicographical order, and (4) the $i$th occurrence of character $c$ in the last column of $M$ corresponds to the $i$th occurrence of character $c$ in the first column, more precisely: If $j$ and $k$ are the positions of the $i$th $c$ in the last and first columns respectively, and the $j$th row of matrix $M$ is $x_1\cdots x_{n-1}x_n$, then the $k$th row is $x_nx_1\cdots x_{n-1}$. This last property can be used to define a mapping from the last to the first column, called {\em LF-mapping}~\cite{BurrowsWheeler94}, which assigns to each position $j$ the corresponding position $k$ in the first column---this is, in fact, the standard permutation of the last column. Now, given $w$ which is the BWT of a word, such a word $v$ can be reconstructed, from last character to first, via iteratively applying the standard permutation $\sigma_w$, and noting that $w_{\sigma_w(i)}$ is the character preceding $w_i$ in $v$. In other words, $w_1=v_n$ and $w_{\sigma_w^i(1)} = v_{n-i}$ for $1\leq i \leq n-1$. 

\subsection{Problem statement and first results} 

Let $w\in \Sigma^*$, $w= w_1\cdots w_n$ and $1\leq i \leq n+1$. We denote by $\dol(w,i)$ the $(n+1)$-length word $w_1 \cdots w_{i-1}\$w_i \cdots w_n$, i.e.\ the word which results from inserting $\$$ into $w$ in position $i$. 
Whenever $w$ is clear from the context, we denote by $\sigma_i$ the standard permutation of $\dol(w,i)$.

\begin{definition}\label{def:nice}
Let $w \in \Sigma^*$, $|w|=n$ and let $1\leq i \leq n+1$. We call position $i$ {\em nice} if $\dol(w,i) \in \BWT(\SigmaD)$, i.e.\ if there exists a word $v\in \Sigma^*$, $|v|=n$, such that $\BWT(v\$)=\dol(w,i)$. 
\end{definition}

We can now state the problem we treat in this paper: 

\medskip

\begin{quote}{\bf Dollar-BWT Problem: }
	Given a word $w \in \Sigma^*$, $|w|=n$, compute all nice positions of $w$.
\end{quote}

\medskip

The following statement was proved for binary alphabets in~\cite{MantaciRS03}, and stated in generalized form for larger alphabets in~\cite{LShur11}: 

\begin{theorem}[Mantaci et al., 2003~\cite{MantaciRS03}, Likhomanov and Shur, 2011~\cite{LShur11}]\label{thm:MRS03}
	For a string $v\in \Sigma^*,$ $\BWT(v) =  a_1^c\cdots a_m^c$ for some $c\geq 1$, if and only if $v = u^c$ with $\BWT(u) = a_1\cdots a_m$. 
\end{theorem}

From this the authors of~\cite{LShur11} obtain the following beautiful result:

\begin{theorem}[Likhomanov and Shur, 2011~\cite{LShur11}] \label{thm:LShur}
	A word $w\in \Sigma^*$ is a  BWT image if and only if the number of cycles of $\sigma_w$ equals the greatest common divisor of the runlengths of $w$. 
\end{theorem}

In the following, we will need the explicit form of the standard permutation of BWT images.

\begin{corollary}\label{coro:formofsigma}
	If $w$ is the BWT of a word $v\in \Sigma^*$ then $\sigma_w$ has the following form, where $c\geq 1$ and $m = n/c$: 
	\begin{equation}
	\label{eq:powerForm} 
	\sigma_w = \powerFormbis.
	\end{equation}
	
	\noindent Moreover, for $c=1$, it holds that $w$ is the BWT of a primitive word if and only if $\sigma_w$ is cyclic. 
\end{corollary}

\begin{proof}
	Let $c$ be the greatest common divisor of the runlengths of $w$. By Thm.~\ref{thm:MRS03}, there exists a primitive word $u$ s.t.\ $v = u^c$. If $c=1$, then, by Thm.~\ref{thm:LShur}, $\sigma_w$ is cyclic, as claimed. Otherwise, let $1\leq i \leq |w|$ and $i-1 = kc +r$, with $0\leq r < c$ be the unique decomposition of $i-1$ modulo $c$. It follows from the definition of the standard permutation that 
	\begin{equation}\label{eq:tauc}
	\sigma_w(i) = \sigma_w(kc+r+1) = \sigma_w(kc+1) + r,  
	\end{equation}
	
	\noindent since $w_i = w_{i-1} = \ldots = w_{kc+1}$, i.e.\ position $i$ and position $kc+1$ lie in the same run. But this implies that the standard permutation of $w$ has the form~\eqref{eq:powerForm}, as claimed. 
	
	For $c=1$, the reverse implication follows from applying the BWT reversal algorithm: if $\sigma_w$ is cyclic, then the output is a word of length $|w|$. Note that this direction is not true for $c>1$: e.g.\ $(1,3)(2,4)$ is the standard permutation of $bbaa$, but also of $cdab$, and the latter is not a BWT image. 

\end{proof}

Let $w\in \Sigma^*$. Since, for every $i$, the character $\$$ appears exactly once in $\dol(w,i)$, from Thm.~\ref{thm:LShur} we immediately get the following: 

\begin{corollary} \label{lemma:iNice}
	For $w\in \Sigma^n$ and $1 \leq i \leq n+1$, $i$ is nice if and only if $\sigma_i$ is cyclic. 
\end{corollary}

We use a bipartite graph $G_w$ to visualize the standard permutation of $w$ (see Fig.~\ref{fig:constraints15}). The top row corresponds to $w$, and the bottom row to the characters of $w$ in alphabetical order. When $w$ is a BWT, then this implies that the top row corresponds to the last column of matrix $M$, and the bottom row to the first. (This graph is therefore sometimes called \BWT-graph.) Let us refer to the nodes in the top row as $x_1,\ldots, x_n$ and to those in the bottom row as $y_1,\ldots, y_n$. Nodes $x_i$ are labeled by character $w_i$, and nodes $y_i$ are labeled by the characters of $w$ in lexicographic order. We connect $(x_i,y_j)$ if and only if $i=j$ or $j = \sigma_w(i)$. It is easy to see that the node set of any cycle $S$ in $G_w$ has the form $\{x_k,y_k \mid k\in {\cal I}\}$ for some ${\cal I} \subseteq \{1,\ldots, n\}$, and that $S$ is a cycle in $G_w$ if and only if ${\cal I}$ is a cycle in $\sigma$.


\begin{figure}
\centering
\caption{\small Standard permutation for $w = {\tt beaaecdcb}$ with $\sigma_w = \bigl( \protect\begin{smallmatrix} 1 & 2 & 3 & 4 & 5 & 6 & 7 & 8 & 9 \\ 3 & 8 & 1 & 2 & 9 & 5 & 7 & 6 & 4 \protect\end{smallmatrix}\bigr)$ and $\sigma_7 =  \bigl( \protect\begin{smallmatrix} 1 & 2 & 3 & 4 & 5 & 6 & 7 & 8 & 9 & 10 \\ 4 & 9 & 2 & 3 & 10 & 6 & 1 & 8 & 7 & 5 \protect\end{smallmatrix}\bigr)$. 
\label{fig:constraints15}}
\subfloat[Standard permutation of $w$]{%
\scalebox{0.44}{%
\begin{tikzpicture}
[thick,
every node/.style={draw,circle,transform shape},
snode/.style={minimum size=0.75cm},
tnode/.style={line width=2.3pt},
every fit/.style={ellipse,draw,inner sep=-2pt,text width=2cm},shorten >= 3pt,shorten <= 3pt
]
    
\begin{scope}[yshift=2cm,start chain=going right,node distance=7mm]

\node[snode,on chain] (s1) [label=above:1] {\Large b};

\node[snode,on chain] (s2) [label=above:2] {\Large e};

\node[snode,on chain] (s3) [label=above:3] {\Large a};

\node[snode,on chain] (s4) [label=above:4] {\Large a};

\node[snode,on chain] (s5) [label=above:5] {\Large e};

\node[snode,on chain] (s6) [label=above:6] {\Large c};

\node[snode,on chain] (s7) [label=above:7] {\Large d};

\node[snode,on chain] (s8) [label=above:8] {\Large c};

\node[snode,on chain] (s9) [label=above:9] {\Large b};
\end{scope}

\begin{scope}[start chain=going right,node distance=7mm]

\node[snode,on chain] (f3) [label=below:1] {\Large a};

\node[snode,on chain] (f4) [label=below:2] {\Large a};

\node[snode,on chain] (f1) [label=below:3] {\Large b};

\node[snode,on chain] (f9) [label=below:4] {\Large b};

\node[snode,on chain] (f6) [label=below:5] {\Large c};

\node[snode,on chain] (f8) [label=below:6] {\Large c};

\node[snode,on chain] (f7) [label=below:7] {\Large d};

\node[snode,on chain] (f2) [label=below:8] {\Large e};

\node[snode,on chain] (f5) [label=below:9] {\Large e};
\end{scope}

    \draw[gray] (s1) -- (f3);
    \draw[gray] (s3) -- (f1);
    \draw[gray] (s2) -- (f4);
    \draw[gray] (s4) -- (f9);
    \draw[gray] (s9) -- (f5);
    \draw[gray] (s5) -- (f6);
    \draw[gray] (s6) -- (f8);
    \draw[gray] (s8) -- (f2);

\draw[gray] (s1) -- (f1);
\draw[gray] (s3) -- (f3);
\draw[gray] (s2) -- (f2);
\draw[gray] (s4) -- (f4);
\draw[gray] (s9) -- (f9);
\draw[gray] (s5) -- (f5);
\draw[gray] (s6) -- (f6);
\draw[gray] (s8) -- (f8);
\draw[gray] (s7) -- (f7);

\end{tikzpicture}}%
}
\qquad
\subfloat[Standard permutation of $\dol(w,7)$]{%
\scalebox{0.41}{%
\begin{tikzpicture}[thick,
every node/.style={draw,circle,transform shape},
snode/.style={minimum size=0.75cm},
tnode/.style={line width=2.3pt},
every fit/.style={ellipse,draw,inner sep=-2pt,text width=2cm},shorten >= 3pt,shorten <= 3pt
]
    
\begin{scope}[yshift=2cm,start chain=going right,node distance=7mm]

\node[snode,on chain] (s1) [label=above:1] {\Large b};

\node[snode,on chain] (s2) [label=above:2] {\Large e};

\node[snode,on chain] (s3) [label=above:3] {\Large a};

\node[snode,on chain] (s4) [label=above:4] {\Large a};

\node[snode,on chain] (s5) [label=above:5] {\Large e};

\node[snode,on chain] (s6) [label=above:6] {\Large c};

\node[snode,on chain] (s7) [label=above:7] {\Large \$};

\node[snode,on chain] (s8) [label=above:8] {\Large d};

\node[snode,on chain] (s9) [label=above:9] {\Large c};

\node[snode,on chain] (s10) [label=above:10] {\Large b};
\end{scope}

\begin{scope}[start chain=going right,node distance=7mm]

\node[snode,on chain] (f7) [label=below:1] {\Large \$};

\node[snode,on chain] (f3) [label=below:2] {\Large a};

\node[snode,on chain] (f4) [label=below:3] {\Large a};

\node[snode,on chain] (f1) [label=below:4] {\Large b};

\node[snode,on chain] (f10) [label=below:5] {\Large b};

\node[snode,on chain] (f6) [label=below:6] {\Large c};

\node[snode,on chain] (f9) [label=below:7] {\Large c};

\node[snode,on chain] (f8) [label=below:8] {\Large d};

\node[snode,on chain] (f2) [label=below:9] {\Large e};

\node[snode,on chain] (f5) [label=below:10] {\Large e};
\end{scope}

    \draw[gray] (s1) -- (f7);
    \draw[gray] (s7) -- (f9);
    \draw[gray] (s9) -- (f2);
    \draw[gray] (s2) -- (f3);
    \draw[gray] (s3) -- (f4);
    \draw[gray] (s4) -- (f1);
    \draw[gray] (s5) -- (f10);
    \draw[gray] (s10) -- (f5);
    \draw[gray] (s6) -- (f6);
    \draw[gray] (s8) -- (f8);

\draw[gray] (s1) -- (f1);
\draw[gray] (s7) -- (f7) [dashed];
\draw[gray] (s9) -- (f9);
\draw[gray] (s2) -- (f2);
\draw[gray] (s3) -- (f3);
\draw[gray] (s4) -- (f4);
\draw[gray] (s5) -- (f5);
\draw[gray] (s10) -- (f10);
\draw[gray] (s6) -- (f6);
\draw[gray] (s8) -- (f8);
\end{tikzpicture}}
}
\end{figure}


Now observe what happens when we insert a dollar into $w$ in position $i$ (see Fig.~\ref{fig:constraints15}). For positions $j$ which are smaller than $i$, their image is incremented by one; $i$ is mapped to $1$; and positions $j$ to the right of $i$,  both $j$ and its image $\sigma(j)$ are shifted to the right by one.  Formally: 

\begin{lemma}\label{lemma:sigma2sigmai}
	Let $w\in \Sigma^n$, $1\leq i \leq n+1$, $\sigma_w$ the standard permutation of $w$, 
	and $\sigma_i$ the standard permutation of $\dol(w,i)$. Then 
	
	\begin{equation*}\label{eq:sigma2sigma_i}
	\sigma_i(j) = 
	\begin{cases}
	\sigma_w(j) + 1 & \text{ if } j<i, \\
	1 & \text{ if } j=i, \text{ and } \\
	\sigma_w(j-1) + 1 & \text{ if } j>i. 
	\end{cases}
	\end{equation*}
	
\end{lemma}

\begin{proof}
	Immediate from the definition. 
\end{proof}

%
%
%
%
%
%

\section{Algorithm}\label{sec:algo}

Given a word $w$, it is easy to compute all nice positions of $w$, by inserting $\$$ in each position $i$ and running the \BWT\ reversal algorithm, in a total of $\Oh(n^2)$ time. Here we present an $\Oh(n \log n)$ time algorithm for the problem. 

The underlying idea is that, if we know $\sigma_i$, the standard permutation of $\dol(w,i)$, then it is not too difficult to compute $\sigma_{i+1}$. 
\begin{lemma}\label{lemma:sigmai}
	Let $w\in \Sigma^n$, and $1 \leq i \leq n$. Then 
	
	\begin{enumerate}
		\item $\sigma_1(1) = 1$ and for $i>1$, $\sigma_1(i) = \sigma_w(i-1)+1$, and 
		\item $\sigma_{i+1} = (1, \sigma_i(i+1)) \cdot \sigma_i$. 
	\end{enumerate}
	
	In particular, the standard permutation $\sigma_{i+1}$ is the result of applying a single transposition to $\sigma_i$. 
	
\end{lemma}

\begin{proof}
	Part 1. follows by applying Lemma~\ref{lemma:sigma2sigmai} to $i=1$. For Part 2., first notice 
	that for all $j\neq i,i+1$, we have $\sigma_i(j) = \sigma_{i+1}(j)$, by Lemma~\ref{lemma:sigma2sigmai}. This is true because $j$ is either smaller than both $i$ and $i+1$, or larger than both $i$ and $i+1$. 
In the first case, i.e.\ if $j <i$, then $\sigma_i(j)=\sigma_w(j)+1=\sigma_{i+1}(j)$. In the second case, i.e.\ if $j>i+1$, then $\sigma_i(j)=\sigma_w(j-1)+1=\sigma_{i+1}(j)$. Finally $\sigma_i(i) = 1 = \sigma_{i+1}(i+1)$, and $\sigma_{i}(i+1) = \sigma_w(i) +1=\sigma_{i+1}(i)$, again by Lemma~\ref{lemma:sigma2sigmai}. 
\end{proof}

As we show next, applying a transposition to a permutation has either the effect of splitting a cycle, or that of merging two cycles. 

\begin{lemma}\label{lemma:cycles}
	Let $\pi = C_1 \cdots C_k$ be the cycle decomposition of the permutation $\pi$, $x\neq y$, and $\pi' = (\pi(x),\pi(y)) \cdot \pi$. 
	\begin{enumerate}
		\item If $x$ and $y$ are in the same cycle $C_i$, then this cycle is split into two. In particular, let $C_i = (c_1, c_2, \ldots, c_j, \ldots, c_m)$, with $c_m = x$ and $c_j = y$. \\Then $\pi' =$ $ (c_1,c_2,\ldots, c_{j-1},y)(c_{j+1} \ldots c_{m-1},x) \prod_{\ell\neq i} C_{\ell}$.  
		
		\item If $x$ and $y$ are in different cycles $C_i$ and $C_j$, then these two cycles are merged. In particular, let $C_i = (c_1, c_2, \ldots, c_m)$, with $c_m = x$, and $C_j = (c'_1, c_2', \ldots, c_r')$, with $c_r'=y$, then $\pi' = (c_1, \ldots, c_{m-1},x, c_1', \ldots, c'_{r-1},y) \prod_{\ell\neq i,j} C_{\ell}$. 
		
	\end{enumerate}
\end{lemma}

\begin{proof} 
	Let $\tau = (\pi(x),\pi(y))$. First note that $\pi'(z) = \tau(\pi(z)) = \pi(z)$ for all $z\neq x,y$. 
	Case {\em 1:} $\pi'(x) = \tau(\pi(x)) = \pi(y) = c_{j+1}$ and $\pi'(y) = \tau(\pi(y)) = \pi(x) = c_1$, which proves the claim. Case 
	{\em 2:} follows analogously. 
	 \end{proof}

See Example~\ref{ex:sigmaPermutations1} for an example.

\subsection{High-level description of the algorithm}

The algorithm first computes the standard permutation $\sigma=\sigma_w$ of $w$  and initializes a counter $c$ with the number of cycles of $\sigma$. It then computes $\sigma_1$ according to Lemma~\ref{lemma:sigmai}, part 1. It increments  counter $c$ by $1$ for $i=1$, since $1$ is always a fixpoint of $\sigma_1$. 
Then the algorithm iteratively computes the new permutation $\sigma_{i+1}$, updating $c$ in each iteration. By Lemma~\ref{lemma:cycles}, $c$ either increases or decreases by $1$ in every iteration: it increases if $i+1$ is in the same cycle as $i$, and it decreases if it is in a different cycle. Whenever $c$ equals $1$, the algorithm reports the current value $i$. See Algorithm~1 for the pseudocode and Example~\ref{ex:sigmaPermutations1} for an example.

\IncMargin{1em}
\begin{algorithm}
	\SetKwProg{Procedure}{procedure}{:}{end}
	\SetKw{DownTo}{down to}
	\DontPrintSemicolon
	{Given a word $w$, return a set $\mathcal{I}$ of positions in which the \$-character can be inserted to turn $w$ into a BWT image.}\;
	
	\LinesNumbered
	\BlankLine
	
	\nl $n \leftarrow |w|$ \;
	\nl $\sigma \gets$ standard permutation of $w$\tcp*[l]{variant of counting sort \hspace{2.2em}}
	$c \leftarrow $ number of cycles of $\sigma$\;
	$\mathcal{I} \leftarrow \emptyset$\;
	\For{$i \leftarrow n+1$ \DownTo $2$}{ 
		$\sigma(i) \leftarrow \sigma(i-1)+1$ \tcp*[l]{\normalfont{compute $\sigma_1$ from $\sigma$}}
	}
	$\sigma(1) \leftarrow 1$\;
	$c \gets c+1$\tcp*[l]{$\sigma_1$ has one more cycle than $\sigma$}

	\For{ $i \leftarrow 1$ \KwTo $n$ }{
		$C \leftarrow$ cycle of $\sigma$ which contains $1$ \tcp*[l]{$C$ also contains $i$}
		\If {$i+1 \in C$}{
			$c \leftarrow c+1$ \tcp*[l]{case split} 
		} 
		\Else{
			$c \leftarrow c-1$ \tcp*[l]{case merge}
		}
		$\sigma \leftarrow$ {\sc Update}$(\sigma, i)$\tcp*[l]{now $\sigma=\sigma_{i+1}$}
		\If{ $c = 1$}{
			$\mathcal{I} \leftarrow \mathcal{I} \cup \{i+1\}$\;
		}
		
	}
	\Return{$\mathcal{I}$}
	\caption{{\sc FindNicePositions}$(w)$\label{algo:findPos}}
	
	\BlankLine
	\Procedure{\sc {\sc Update}($\sigma, i$)}{ \label{procedure}
		
		$\sigma(i) \leftarrow \sigma(i+1)$ \tcp*[l]{$\sigma(i)=1$}
		$\sigma(i+1) \leftarrow 1$\;
		
		\Return{$\sigma$\;}
		
	}
	
\end{algorithm}
\DecMargin{1em}


\begin{exa}\label{ex:sigmaPermutations1} 

Let $w = {\tt  acccbccbab}$. 
We report $\sigma=\sigma_w$, followed by $\sigma_i$ for $i=1, \ldots, 11$,  both in two-line and cyclic representations. Changes from one permutation to the next  are highlighted in red (bold grey in the black-and-white print version), and cyclic $\sigma_i$, i.e.\ nice positions $i$, are marked with a box. On the right, we note whether a merge or a split has taken place. By Lemma~\ref{lemma:cycles}, this depends on whether $i$ and $i+1$ are in distinct cycles or in the same cycle. In the last column, we report the number of cycles of $\sigma_i$. 
\[w = {\tt  acccbccbab} \qquad \sigma = \bigl( \begin{smallmatrix} 1 & 2 & 3 & 4 & 5 & 6 & 7 & 8 & 9 & 10 \\ 1 & 6 & 7 & 8 & 3 & 9 & 10 & 4 & 2 & 5 \end{smallmatrix} \bigr) = (1)(2, 6, 9)(3, 7, 10, 5)(4, 8)\]

\centering
\begin{minipage}[t]{0.8\textwidth}
	
	\vspace{0.5cm}
	
$\sigma_1 = \bigl( \begin{smallmatrix} \fat{\red{1}} & 2 & 3 & 4 & 5 & 6 & 7 & 8 & 9 & 10 & 11 \\ \fat{\red{1}} & 2 & 7 & 8 & 9 & 4 & 10 & 11 & 5 & 3 & 6 \end{smallmatrix} \bigr) = (\fat{\red{1}})(\fat{\red{2}})(3, 7, 10)(4, 8, 11, 6)(5, 9)$
\bigskip

$\sigma_2 = \bigl( \begin{smallmatrix} 1 & 2 & 3 & 4 & 5 & 6 & 7 & 8 & 9 & 10 & 11 \\ \fat{\red{2}} & \fat{\red{1}} & 7 & 8 & 9 & 4 & 10 & 11 & 5 & 3 & 6 \end{smallmatrix} \bigr) = (1, \fat{\red{2}})(\fat{\red{3}}, 7, 10)(4, 8, 11, 6)(5, 9)$
\bigskip

$\sigma_3 = \bigl( \begin{smallmatrix} 1 & 2 & 3 & 4 & 5 & 6 & 7 & 8 & 9 & 10 & 11 \\ 2 & \fat{\red{7}} & \fat{\red{1}} & 8 & 9 & 4 & 10 & 11 & 5 & 3 & 6 \end{smallmatrix} \bigr) = (1, 2, 7, 10, \fat{\red{3}})(\fat{\red{4}}, 8, 11, 6)(5, 9)$
\bigskip

$\sigma_4 = \bigl( \begin{smallmatrix} 1 & 2 & 3 & 4 & 5 & 6 & 7 & 8 & 9 & 10 & 11 \\  2 & 7 & \fat{\red{8}} & \fat{\red{1}} & 9 & 4 & 10 & 11 & 5 & 3 & 6 \end{smallmatrix} \bigr) = (1, 2, 7, 10, 3, 8, 11, 6, \fat{\red{4}})(\fat{\red{5}}, 9)$
\bigskip

\fbox{$\sigma_5 = \bigl( \begin{smallmatrix} 1 & 2 & 3 & 4 & 5 & 6 & 7 & 8 & 9 & 10 & 11 \\ 2 & 7 & 8 & \fat{\red{9}} & \fat{\red{1}} & 4 & 10 & 11 & 5 & 3 & 6 \end{smallmatrix} \bigr) = (1, 2, 7, 10, 3, 8, 11, \fat{\red{6}}, 4, 9, \fat{\red{5}})$} 
\bigskip

$\sigma_6 = \bigl( \begin{smallmatrix} 1 & 2 & 3 & 4 & 5 & 6 & 7 & 8 & 9 & 10 & 11 \\ 2 & 7 & 8 & 9 & \fat{\red{4}} & \fat{\red{1}} & 10 & 11 & 5 & 3 & 6 \end{smallmatrix} \bigr) = (1, 2, \fat{\red{7}}, 10, 3, 8, 11, \fat{\red{6}})(4, 9, 5)$
\bigskip

$\sigma_7 = \bigl( \begin{smallmatrix} 1 & 2 & 3 & 4 & 5 & 6 & 7 & 8 & 9 & 10 & 11 \\ 2 & 7 & 8 & 9 & 4 & \fat{\red{10}} & \fat{\red{1}} & 11 & 5 & 3 & 6 \end{smallmatrix} \bigr) = (1, 2, \fat{\red{7}})(10, 3, \fat{\red{8}}, 11, 6)(4, 9, 5)$
\bigskip

$\sigma_8 = \bigl( \begin{smallmatrix} 1 & 2 & 3 & 4 & 5 & 6 & 7 & 8 & 9 & 10 & 11 \\ 2 & 7 & 8 & 9 & 4 & 10 & \fat{\red{11}} & \fat{\red{1}} & 5 & 3 & 6 \end{smallmatrix} \bigr) = (1, 2, 7, 11, 6, 10, 3, \fat{\red{8}})(4, \fat{\red{9}}, 5)$
\bigskip

\fbox{$\sigma_9 = \bigl( \begin{smallmatrix} 1 & 2 & 3 & 4 & 5 & 6 & 7 & 8 & 9 & 10 & 11 \\ 2 & 7 & 8 & 9 & 4 & 10 & 11 & \fat{\red{5}} & \fat{\red{1}} & 3 & 6 \end{smallmatrix} \bigr) = (1, 2, 7, 11, 6, \fat{\red{10}}, 3, 8, 5, 4, \fat{\red{9}})$}
\bigskip

$\sigma_{10} = \bigl( \begin{smallmatrix} 1 & 2 & 3 & 4 & 5 & 6 & 7 & 8 & 9 & 10 & 11 \\ 2 & 7 & 8 & 9 & 4 & 10 & 11 & 5 & \fat{\red{3}} & \fat{\red{1}} & 6 \end{smallmatrix} \bigr) = (1, 2, 7, \fat{\red{11}}, 6, \fat{\red{10}})(3, 8, 5, 4, 9)$
\bigskip

$\sigma_{11} = \bigl( \begin{smallmatrix} 1 & 2 & 3 & 4 & 5 & 6 & 7 & 8 & 9 & 10 & 11 \\ 2 & 7 & 8 & 9 & 4 & 10 & 11 & 5 & 3 & \fat{\red{6}} & \fat{\red{1}} \end{smallmatrix} \bigr) = (1, 2, 7, 11)(6, 10)(3, 8, 5, 4, 9)$

\bigskip
\end{minipage}
\hspace{-2.5cm}
\begin{minipage}[t]{0.2\textwidth}
	
	\vspace{1.15cm}
	
	{\em merge}
	\bigskip
	
	{\em merge}
	\bigskip
	
	{\em merge}
	\bigskip
	
	{\em merge}
	\bigskip
	\vspace{0.3cm}
	
	{\em split}
	\bigskip
	
	{\em split}
	\bigskip
	
	{\em merge}
	\bigskip
	
	{\em merge}
	\bigskip
	\vspace{0.3cm}
	
	{\em split}
	\bigskip
	
	{\em split}
	
\end{minipage}
\hspace{-1.5cm}\begin{minipage}[t]{0.05\textwidth}
	
	\hspace{-1.2cm}{\bf no.\ cycles}
	\bigskip
	
	5
	\bigskip
	
	4
	\bigskip
	
	3
	\bigskip
	
	2
	\vspace{0.55cm}
	
	1
	\vspace{0.55cm}
	
	2
	\bigskip
	
	3
	\bigskip
	
	2
	\vspace{0.55cm}
	
	1
	\vspace{0.55cm}
	
	2
	\bigskip
	
	3
\end{minipage}

\end{exa}



\subsection{Implementation with splay trees}

For the algorithm's implementation, we need an appropriate data structure for maintaining and updating the current permutation $\sigma_i$. Using an array to keep $\sigma_i$ would allow us to update it in constant time in each step, but would not give us the possibility to efficiently decide whether $i+1\in C$ in line $11$. Thus we need a data structure to maintain the cycles of $\sigma_i$. The functionalities we seek are (a) decide whether two elements are in the same cycle, (b) split two cycles, (c) merge two cycles. The data structure we have chosen is a forest of splay trees~{\cite{SleatorT85}. This data structure supports the above operations in amortized $\Oh(\log n)$ time. 

Splay trees are self-adjusting binary search trees. They are not necessarily balanced, but they have the property that at every update or access-operation, the element $x$ accessed is moved to the root and the tree is adjusted in such a way as to reduce future access-time to nodes on the path from the root to $x$. The basic operation (moving a node to the root) is called {\em splaying} and consists of a series of the usual edge rotations in binary search trees. Which rotations are applied depends on the position of the node with respect to its parent and grandparent (the cases are referred to as {\em zig}, {\em zig-zig}, and {\em zig-zag}). Splay trees can implement the standard operations on binary search trees, such as {\em access, insert, delete, join, split} in amortized logarithmic time, in the total number of nodes involved. We refer the reader to the original article~\cite{SleatorT85} for more details. 

We will represent the current permutation $\sigma_i$ as a forest of splay trees, where each tree corresponds to a cycle of $\sigma_i$. Let $(c_1, c_2, \ldots, c_k)$ be an arbitrary rotation of a cycle in $\sigma_i$. We consider the cycle as a ranked list of the elements from $c_1$ to $c_k$ and assign to element $c_j$ its position $j$ as key. The trick is that we do not explicitly store $j$, but will only implicitly use the binary search tree property: namely, that the inorder traversal of the tree yields the elements stored in the nodes in the correct order. To access elements, we store pointers to the corresponding nodes; each  access is followed by splaying the node. This will allow us to implement the split- and merge-operations we defined on cycles, using only basic splay-tree operations. 

\bigskip

We now explain how to update the data structure, i.e.\ how to go from $\sigma_i$ to $\sigma_{i+1}$. %

Recall that we have to split a cycle or merge two cycles, depending on whether $i+1$ is in the same cycle as $i$ (Lemma~\ref{lemma:cycles}). In order to be able to check this, we maintain a pointer to the root of the first cycle. To check whether $i+1$ is in the first cycle (the cycle containing $i$), we splay $i+1$; if the root of the first cycle now has a parent, then the answer is {\sc yes}, and $i+1$ becomes the new root, otherwise the answer is {\sc no}. 

If $i+1$ is in the first cycle, then the first cycle has to be split after $i+1$. Let the cycle have the form $(A, i+1, B)$, with $A,B$ sequences of numbers.  (Recall that the first element of the first cycle is $1$ and the last element is  $i$.) 
The corresponding splay tree after splaying $i+1$ is shown in Fig.~\ref{fig:abstractSplit}, part (a). The {\em split} operation cuts the right subtree of $i+1$,  producing the two new trees in Fig.~\ref{fig:abstractSplit}, part (b), with the resulting cycles $(A, i+1)$ and $(B)$. 


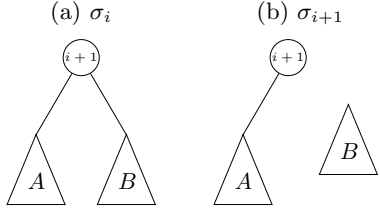
\begin{figure}
\centering
\caption{\label{fig:abstractSplit}The implementation of the split-step with splay trees.\label{fig:abstractSplit}}

\subfloat[$\sigma_i$]{%
\scalebox{0.80}{%
\begin{tikzpicture}
[level distance=1.5cm,
  level 1/.style={sibling distance=1.5cm},
  level 2/.style={sibling distance=1cm}]
\node[node, minimum size=6mm, label=center:{\tiny $i+1$}] (i1) at (0,0) { }
  	child[white]{node(A)[]{}}
  	child[white]{node(B)[]{}};

\node[shape border rotate=90, yshift=-0.68cm, isosceles triangle, fit = (A), draw]{$A$};
\node[shape border rotate=90, yshift=-0.68cm, isosceles triangle, fit = (B), draw]{$B$};

\draw (i1) -- (-.76,-1.3);
\draw (i1) -- (.76,-1.3);

\end{tikzpicture}}%
}
\qquad
\subfloat[$\sigma_{i+1}$]{%
\scalebox{0.80}{%
\begin{tikzpicture}[level distance=1.5cm,
  level 1/.style={sibling distance=1.5cm},
  level 2/.style={sibling distance=1cm}]
		
\node[ node, minimum size=6mm, label=center:{\tiny $i+1$}] at (0,0) { }
  	child[white]{node(A)[]{}}
  	child[missing]{node[]{}};
  			
\node[] (B) at (1,-1){};  			
  			
\node[shape border rotate=90, yshift=-0.68cm, isosceles triangle, fit = (A), draw]{$A$};
\node[shape border rotate=90, yshift=-0.68cm, isosceles triangle, fit = (B), draw]{$B$};

\draw (i1) -- (-.76,-1.3);

\end{tikzpicture}}
}
\end{figure}


If $i+1$ is in a different cycle, then these two cycles have to be merged, according to Lemma~\ref{lemma:cycles}. Let the two cycles have the form $(A, i)$ and $(B, i+1, D)$, respectively, with $A,B,D$ sequences of numbers, see Fig.~\ref{fig:abstractMerge}, part (a).  (Recall that the first element of $A$ is necessarily $1$.) The new cycle, at the end of the merge-operation, will then be $(A, i, D, B, i+1)$. First we {\em split} the right subtree of $i+1$ (which is already at the root), resulting in two trees, one representing $(B,i+1)$ and the other, $(D)$. Let $x$ be the rightmost element of $D$, and $D'$ be $D$ without its last element, i.e.\ $(D) = (D', x)$. Next we splay both $i$ and $x$. Altogether, we now have three trees, each with its rightmost element at its root: one representing $(A,i)$, one $(B,i+1)$, and one $(D',x)$, see Fig.~\ref{fig:abstractMerge}, part (b). We now perform two {\em join} operations, joining the third tree as right child of $i$, and the second as right child of $x$. We thus get a complete tree respresenting $(A, i, D', x, B, i+1) = (A, i, D, B, i+1)$, as desired, see Fig.~\ref{fig:abstractMerge}, part (c). We refer the reader to Fig.~\ref{fig:sigmaPermutations1}  (in the Appendix) for a full example. 


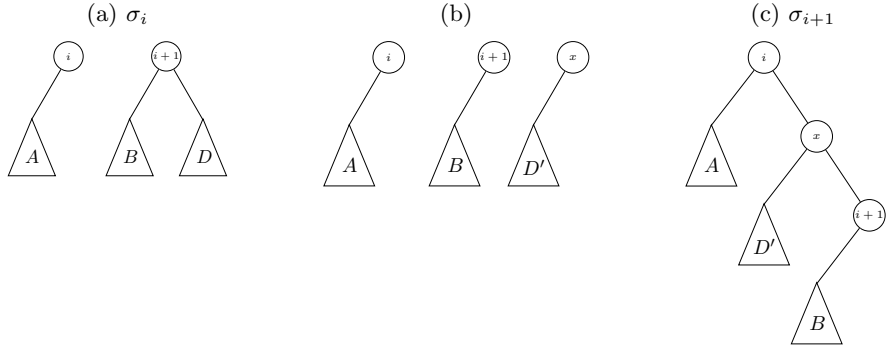
\begin{figure}[h!]
\captionsetup[subfigure]{font=footnotesize}
\centering
\caption{The implementation of the merge-step with splay trees.\label{fig:abstractMerge}}

\subfloat[$\sigma_i$]{%
\scalebox{0.65}{%
\begin{tikzpicture}[level distance=1.5cm,
  level 1/.style={sibling distance=1.5cm},
  level 2/.style={sibling distance=1cm}]
\node[node, minimum size=6mm, label=center:{\tiny $i$}] (i) at (0,0) {}
  	child[white]{node(A)[]{}}
  	child[missing]{node[]{}};

\node[node, minimum size=6mm, label=center:{\tiny $i+1$}] (i1) at (2,0) {}
  	child[white]{node(B)[]{}}
  	child[white]{node(D)[]{}};
  			
\node[shape border rotate=90, yshift=-0.68cm, isosceles triangle, fit = (A), draw]{$A$};
\node[shape border rotate=90, yshift=-0.68cm, isosceles triangle, fit = (B), draw]{$B$};
\node[shape border rotate=90, yshift=-0.68cm, isosceles triangle, fit = (D), draw]{$D$};

\draw (i) -- (-0.76,-1.3);

\draw (i1) -- (1.24,-1.3);
\draw (i1) -- (2.75,-1.3);

\end{tikzpicture}}%
}
\hspace{3em}
\subfloat[]{%
\scalebox{0.7}{%
\begin{tikzpicture}[level distance=1.5cm,
  level 1/.style={sibling distance=1.5cm},
  level 2/.style={sibling distance=1cm}]

\node[node, minimum size=6mm, label=center:{\tiny $i$}] (i) at (0,0) { }
  	child[white]{node(A)[]{}}
  	child[missing]{node[]{}};

\node[node, minimum size=6mm, label=center:{\tiny $i+1$}] (i1) at (2,0) { }
  	child[white]{node(B)[]{}}
  	child[missing]{node[]{}};
  			
\node[node, minimum size=6mm, label=center:{\tiny $x$}] (x) at (3.5,0) { }
child[white]{node(D1)[]{}}
child[missing]{node[]{}};
  
\node[shape border rotate=90, yshift=-0.68cm, isosceles triangle, fit = (A), draw]{$A$};		
\node[shape border rotate=90, yshift=-0.68cm, isosceles triangle, fit = (B), draw]{$B$};
\node[shape border rotate=90, yshift=-0.68cm, isosceles triangle, fit = (D1), draw, label={[label distance=-0.75cm]30:$D'$}]{};

\draw (i) -- (-0.76,-1.3);

\draw (i1) -- (1.24,-1.3);

\draw (x) -- (2.74,-1.3);
\end{tikzpicture}}%
}
\hspace{3em}
\subfloat[$\sigma_{i+1}$]{%
\scalebox{0.7}{%
\begin{tikzpicture}[level distance=1.5cm,
  level 1/.style={sibling distance=2cm},
  level 2/.style={sibling distance=2cm}]
\node[node, minimum size=6mm, label=center:{\tiny $i$}] (i) at (2,0) { }
  	child[white]{node(A)[]{}}
  	child{node[node, minimum size=6mm, label=center:{\tiny $x$}](x){}
  		child[white]{node(D1)[]{}}
  		child{node[node, minimum size=6mm, label=center:{\tiny $i+1$}] (i1){}
  			child[white]{node(B)[]{}}
  			child[missing]{}}};

\node[shape border rotate=90, yshift=-0.68cm, isosceles triangle, fit = (A), draw]{$A$};
\node[shape border rotate=90, yshift=-0.68cm, isosceles triangle, fit = (B), draw]{$B$};
\node[shape border rotate=90, yshift=-0.68cm, isosceles triangle, fit = (D1), draw, label={[label distance=-0.75cm]30:$D'$}]{};

\draw (i) -- (0.99,-1.3);

\draw (x) -- (2,-2.79);

\draw (i1) -- (3,-4.3);
\end{tikzpicture}}
}
\end{figure}



\subsection{Analysis}

We now show that Algorithm~\ref{algo:findPos} takes $\Oh(n \log n)$ time in the worst case. 

Computing the standard permutation of $w$ takes $\Oh(n)$ time (using a variant of Counting Sort~\cite{Cormen}, and noting that the alphabet of the string has cardinality at most $n$). The computation of $\sigma_1$ (lines $5$ to $7$) takes $\Oh(n)$ time. All steps in one iteration of the for-loop (lines $9$ to $17$) take constant time, except deciding whether $i+1 \in C$ (line $11$), and updating $\sigma$ (line $15$). 
For deciding whether $i+1 \in C$, we {\em access} $i+1$. If the answer is {\sc yes}, we will have a split-step: this is a {\em split}-operation on the tree for $C$ (Fig.~\ref{fig:abstractSplit}). If the answer is {\sc no}, then we merge the two trees; the implementation of this consists in two {\em access}-, one {\em split}-, and two {\em join}-operations on the trees  (Fig.~\ref{fig:abstractMerge}). 

Therefore, in one iteration of the for-loop, we either have one {\em access} and one {\em split} operation (for a split-step), or two {\em access}-, one {\em split}-, and two {\em join}-operations (merge-step), thus in either case, at most five splay tree operations. There are $n$ iterations of the for-loop, so at most $5n$ operations. Together with the initial insertion of the $n+1$ nodes, we get a total of at most $6n+1$ operations. We report the relevant theorem from~\cite{SleatorT85}:

\begin{theorem} [Balance Theorem with Updates, Thm. 6 in~\cite{SleatorT85}] \label{thm:SleatorT85}
	A sequence of $m$ arbitrary operations on a collection of initially empty splay trees takes $\mathcal{O}(m+ \sum^{m}_{j=1} \log n_j)$ time, where $n_j$ is the number of items in the tree or trees involved in operation $j$.
\end{theorem}

For our algorithm, we have $m=\Oh(n)$ operations altogether, each involving no more than $n+1$ nodes, thus Theorem~\ref{thm:SleatorT85} guarantees that the total time spent on the splay trees is $\Oh(n + n \log n)$. Adding to this the computation of $\sigma_1$ in $\Oh(n)$ time, and of the constant-time operations within the for-loop, we get altogether $\Oh(n \log n)$ time. Memory usage is $\Oh(n)$, since the forest of splay trees consists of $n+1$ vertices in total. 
Summarizing, we have 

\begin{theorem}\label{thm:algoanalysis}
	Algorithm~\ref{algo:findPos} runs in  $\Oh(n \log n)$ time and uses  $\Oh(n)$ space, for an input string of length $n$. 
\end{theorem}

%
%

\section{Characterization} \label{sec:characterization}

In this section we give a full characterization of nice positions. 
Our first result is that every BWT image has at least one nice position. 

\begin{theorem} \label{thm:dol(BWT,c+1)}	
	Let $w\in \BWT(\Sigma^n)$, and $c$ the number of cycles of $\sigma_w$. Then $c+1$ is nice. 
\end{theorem}

\begin{proof}
	By Corollary~\ref{coro:formofsigma}, $\sigma_w$ has the form 
	\[ \sigma_w = \powerFormbis.\]

	By Lemma~\ref{lemma:sigma2sigmai}, $\sigma_{c+1}(c+1)=1$. Note that each cycle of $\sigma_w$ has exactly one element which is smaller than $c+1$, therefore for exactly one element $j$ of each cycle, $\sigma_{c+1}(j)=\sigma_w(j)+1$, again by Lemma~\ref{lemma:sigma2sigmai}. These elements are $\sigma_{c+1}(1)=e_1+1, \sigma_{c+1}(2)=e_1+2, \ldots, \sigma_{c+1}(c)=e_1+c$. For all other elements $j>c+1$, it holds that $\sigma_{c+1}(j) = \sigma_w(j-1)+1$. In particular, $\sigma_{c+1}(e_1+1) = \sigma_w(e_1)+1=e_2+1, \sigma_{c+1}(e_2+1)=\sigma_w(e_2)+1=e_3+1, \ldots, \sigma_{c+1}(e_i+c)=\sigma_w(e_i+c-1)+1=e_{i+1}+c$. The resulting form of $\sigma_{c+1}$ is thus	
	\[ \sigma_{c+1} = (1, e_1+1, \ldots, e_m+1, 2, e_1+2, \ldots, e_m+2, \ldots, c, e_1+c, \ldots, e_m+c, c+1),\]
	
	\noindent and is thus cyclic, and therefore, by Corollary~\ref{coro:formofsigma}, $c+1$ is nice.

\end{proof}

\begin{proposition}
	Let $w = \BWT(v\$)$ such that $w_2 = \$$, i.e. \$ is in the second position of $w$. Then $v$ is Lyndon.
\end{proposition}

\begin{proof}
	Let $v' = v\$$. The smallest rotation of $v'$ is $\$ v$, since $\$$ is smaller than all other characters; while $v\$$ is the second smallest rotation, since $w_2 = \$$. Therefore, every proper suffix $u$ of $v$ is lexicographically strictly larger than $v$, implying that $v$ is Lyndon. 
	
\end{proof}

In order to see which positions are nice, we want to understand how cycles in $\sigma_i$ are created. Recall our earlier example $w = {\tt beaaecdcb}$ and $i=7$ (see Fig.~\ref{fig:constraints15FIXPOINTS}). Position $7$ is not nice, since $\sigma_7$ has two fixpoints. 


\begin{figure}
\centering
\caption{\small Standard permutation for $w = {\tt beaaecdcb}$ with $\sigma_w = \bigl( \protect\begin{smallmatrix} 1 & 2 & 3 & 4 & 5 & 6 & 7 & 8 & 9 \\ 3 & 8 & 1 & 2 & 9 & 5 & 7 & 6 & 4 \protect\end{smallmatrix}\bigr)$ and $\sigma_7 =  \bigl( \protect\begin{smallmatrix} 1 & 2 & 3 & 4 & 5 & 6 & 7 & 8 & 9 & 10 \\ 4 & 9 & 2 & 3 & 10 & 6 & 1 & 8 & 7 & 5 \protect\end{smallmatrix}\bigr)$. Red edges (bold grey in the black-and-white print version) become fixpoints in $\sigma_7$. 
\label{fig:constraints15FIXPOINTS}}
\subfloat[Standard permutation of $w$]{%
\scalebox{0.44}{%
\begin{tikzpicture}
[thick,
every node/.style={draw,circle,transform shape},
snode/.style={minimum size=0.75cm},
tnode/.style={line width=2.3pt},
every fit/.style={ellipse,draw,inner sep=-2pt,text width=2cm},shorten >= 3pt,shorten <= 3pt
]
    
\begin{scope}[yshift=2cm,start chain=going right,node distance=7mm]

\node[snode,on chain] (s1) [label=above:1] {\Large b};

\node[snode,on chain] (s2) [label=above:2] {\Large e};

\node[snode,on chain] (s3) [label=above:3] {\Large a};

\node[snode,on chain] (s4) [label=above:4] {\Large a};

\node[snode,on chain] (s5) [label=above:5] {\Large e};

\node[tnode,on chain,red] (s6) [label=above:6] {\Large c};

\node[tnode,on chain,red] (s7) [label=above:7] {\Large d};

\node[snode,on chain] (s8) [label=above:8] {\Large c};

\node[snode,on chain] (s9) [label=above:9] {\Large b};
\end{scope}

\begin{scope}[start chain=going right,node distance=7mm]

\node[snode,on chain] (f3) [label=below:1] {\Large a};

\node[snode,on chain] (f4) [label=below:2] {\Large a};

\node[snode,on chain] (f1) [label=below:3] {\Large b};

\node[snode,on chain] (f9) [label=below:4] {\Large b};

\node[tnode,on chain,red] (f6) [label=below:5] {\Large c};

\node[snode,on chain] (f8) [label=below:6] {\Large c};

\node[tnode,on chain,red] (f7) [label=below:7] {\Large d};

\node[snode,on chain] (f2) [label=below:8] {\Large e};

\node[snode,on chain] (f5) [label=below:9] {\Large e};
\end{scope}

    \draw[gray] (s1) -- (f3);
    \draw[gray] (s3) -- (f1);
    \draw[gray] (s2) -- (f4);
    \draw[gray] (s4) -- (f9);
    \draw[gray] (s9) -- (f5);
    \draw[gray] (s5) -- (f6);
    \draw[gray] (s6) -- (f8);
    \draw[gray] (s8) -- (f2);

\draw[gray] (s1) -- (f1);
\draw[gray] (s3) -- (f3);
\draw[gray] (s2) -- (f2);
\draw[gray] (s4) -- (f4);
\draw[gray] (s9) -- (f9);
\draw[gray] (s5) -- (f5);
\draw[red,line width=2.3pt] (s6) -- (f6);
\draw[gray] (s8) -- (f8);
\draw[red,line width=2.3pt] (s7) -- (f7);

\end{tikzpicture}}%
}
\qquad
\subfloat[Standard permutation of $\dol(w,7)$]{%
\scalebox{0.41}{%
\begin{tikzpicture}[thick,
every node/.style={draw,circle,transform shape},
snode/.style={minimum size=0.75cm},
tnode/.style={line width=2.3pt},
every fit/.style={ellipse,draw,inner sep=-2pt,text width=2cm},shorten >= 3pt,shorten <= 3pt
]
    
\begin{scope}[yshift=2cm,start chain=going right,node distance=7mm]

\node[snode,on chain] (s1) [label=above:1] {\Large b};

\node[snode,on chain] (s2) [label=above:2] {\Large e};

\node[snode,on chain] (s3) [label=above:3] {\Large a};

\node[snode,on chain] (s4) [label=above:4] {\Large a};

\node[snode,on chain] (s5) [label=above:5] {\Large e};

\node[tnode,on chain,red] (s6) [label=above:6] {\Large c};

\node[snode,on chain] (s7) [label=above:7] {\Large \$};

\node[tnode,on chain,red] (s8) [label=above:8] {\Large d};

\node[snode,on chain] (s9) [label=above:9] {\Large c};

\node[snode,on chain] (s10) [label=above:10] {\Large b};
\end{scope}

\begin{scope}[start chain=going right,node distance=7mm]

\node[snode,on chain] (f7) [label=below:1] {\Large \$};

\node[snode,on chain] (f3) [label=below:2] {\Large a};

\node[snode,on chain] (f4) [label=below:3] {\Large a};

\node[snode,on chain] (f1) [label=below:4] {\Large b};

\node[snode,on chain] (f10) [label=below:5] {\Large b};

\node[tnode,on chain,red] (f6) [label=below:6] {\Large c};

\node[snode,on chain] (f9) [label=below:7] {\Large c};

\node[tnode,on chain,red] (f8) [label=below:8] {\Large d};

\node[snode,on chain] (f2) [label=below:9] {\Large e};

\node[snode,on chain] (f5) [label=below:10] {\Large e};
\end{scope}

    \draw[gray] (s1) -- (f7);
    \draw[gray] (s7) -- (f9);
    \draw[gray] (s9) -- (f2);
    \draw[gray] (s2) -- (f3);
    \draw[gray] (s3) -- (f4);
    \draw[gray] (s4) -- (f1);
    \draw[gray] (s5) -- (f10);
    \draw[gray] (s10) -- (f5);
    \draw[gray] (s6) -- (f6);
    \draw[gray] (s8) -- (f8);

\draw[gray] (s1) -- (f1);
\draw[gray] (s7) -- (f7) [dashed];
\draw[gray] (s9) -- (f9);
\draw[gray] (s2) -- (f2);
\draw[gray] (s3) -- (f3);
\draw[gray] (s4) -- (f4);
\draw[gray] (s5) -- (f5);
\draw[gray] (s10) -- (f10);
\draw[red,line width=2.3pt] (s6) -- (f6);
\draw[red,line width=2.3pt] (s8) -- (f8);
\end{tikzpicture}}
}
\end{figure}
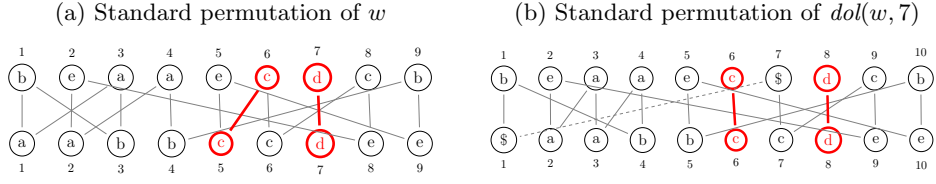


In general, if $j$ is a fixpoint in $\sigma_w$ and $i \leq j$, then $j+1$ will be a fixpoint in $\sigma_i$. Similarly, if $\sigma(j) = j-1$ and $j<i$, then $j$ is a fixpoint in $\sigma_i$. These two cases are illustrated in Fig.~\ref{fig:constraints15FIXPOINTS}, where $7$ is a fixpoint in $\sigma_w$ and $\sigma_w(6)=5$, so the insertion of $\$$ in position $i=7$  leads to the two fixpoints $6$ and $8$ in $\sigma_7$. Therefore, position $7$ is not nice. 

Indeed, the previous observation can be generalized: if $S$ is a cycle in $\sigma_w$, then no position $i\leq \min S$ is nice.
Similarly, if $S$ is such that $\sigma_w(S) = S-1 = \{ j-1 \mid j\in S\}$, then no position $i>\max S$ is nice; in both cases, the insertion of $\$$ in such a position would turn $S$ into a cycle. However, the situation can also be more complex, as is illustrated in Fig.~\ref{fig:cedcbbabb}. In Theorem~\ref{thm:pseudocycle} we will give a necessary and sufficient condition for creating a proper cycle by inserting a \$ in some position. First we need another definition.


\begin{figure}
\centering
\caption{Standard permutation of $w = {\tt cedcbbabb}$ with 
$\sigma = \bigl( \protect\begin{smallmatrix} 1 & 2 & 3 & 4 & 5 & 6 & 7 & 8 & 9 \\ 6 & 9 & 8 & 7 & 2 & 3 & 1 & 4 & 5 \protect\end{smallmatrix}\bigr)$ and of $\dol(w,7)$ with $\sigma_7 =  \bigl( \protect\begin{smallmatrix} 1 & 2 & 3 & 4 & 5 & 6 & 7 & 8 & 9 & 10 \\ 7 & 10 & 9 & 8 & 3 & 4 & 1 & 2 & 5 & 6 \protect\end{smallmatrix}\bigr).$
\label{fig:cedcbbabb}}

\subfloat[Pseudo-cycle in the standard permutation of $w$\label{fig:cedcbbabb-a}]{%
\scalebox{0.42}{%
\begin{tikzpicture}
[thick,
every node/.style={draw,circle,transform shape},
snode/.style={minimum size=0.75cm},
tnode/.style={line width=2.3pt},
every fit/.style={ellipse,draw,inner sep=-2pt,text width=2cm},shorten >= 3pt,shorten <= 3pt
]
    
\begin{scope}[yshift=2cm,start chain=going right,node distance=7mm]

\node[snode,on chain] (s6) [label=above:1] {\Large c};

\node[snode,on chain] (s9) [label=above:2] {\Large e};

\node[tnode,on chain,red] (s8) [label=above:3] {\Large d};

\node[snode,on chain] (s7) [label=above:4] {\Large c};

\node[tnode,on chain,red] (s2) [label=above:5] {\Large b};

\node[snode,on chain] (s3) [label=above:6] {\Large b};

\node[snode,on chain] (s1) [label=above:7] {\Large a};

\node[tnode,on chain,red] (s4) [label=above:8] {\Large b};

\node[snode,on chain] (s5) [label=above:9] {\Large b};
\end{scope}

\begin{scope}[start chain=going right,node distance=7mm]

\node[snode,on chain] (f1) [label=below:1] {\Large a};

\node[tnode,on chain,red] (f2) [label=below:2] {\Large b};

\node[snode,on chain] (f3) [label=below:3] {\Large b};

\node[tnode,on chain,red] (f4) [label=below:4] {\Large b};

\node[snode,on chain] (f5) [label=below:5] {\Large b};

\node[snode,on chain] (f6) [label=below:6] {\Large c};

\node[snode,on chain] (f7) [label=below:7] {\Large c};

\node[tnode,on chain,red] (f8) [label=below:8] {\Large d};

\node[snode,on chain] (f9) [label=below:9] {\Large e};
\end{scope}

    \draw[gray] (s1) -- (f7);
    \draw[gray] (s7) -- (f4);
    \draw[gray] (s4) -- (f8);
    \draw[gray] (s8) -- (f3);
    \draw[gray] (s3) -- (f6);
    \draw[gray] (s6) -- (f1);
    \draw[gray] (s2) -- (f5);
    \draw[gray] (s5) -- (f9);
    \draw[gray] (s9) -- (f2);

\draw[gray] (s1) -- (f1);
\draw[gray] (s7) -- (f7);
\draw[red,line width=2.3pt] (s4) -- (f4);
\draw[red,line width=2.3pt] (s8) -- (f8);
\draw[gray] (s3) -- (f3);
\draw[gray] (s6) -- (f6);
\draw[red,line width=2.3pt] (s2) -- (f2);
\draw[gray] (s5) -- (f5);
\draw[gray] (s9) -- (f9);

\end{tikzpicture}}%
}
\qquad
\subfloat[Cycle in the standard permutation of $\dol(w,7)$\label{fig:cedcbbabb-b}]{%
\scalebox{0.41}{%
\begin{tikzpicture}[thick,
every node/.style={draw,circle,transform shape},
snode/.style={minimum size=0.75cm},
tnode/.style={line width=2.3pt},
every fit/.style={ellipse,draw,inner sep=-2pt,text width=2cm},shorten >= 3pt,shorten <= 3pt
]
    
\begin{scope}[yshift=2cm,start chain=going right,node distance=7mm]

\node[snode,on chain] (s7) [label=above:1] {\Large c};

\node[snode,on chain] (s10) [label=above:2] {\Large e};

\node[tnode,on chain,red] (s9) [label=above:3] {\Large d};

\node[snode,on chain] (s8) [label=above:4] {\Large c};

\node[tnode,on chain,red] (s3) [label=above:5] {\Large b};

\node[snode,on chain] (s4) [label=above:6] {\Large b};

\node[snode,on chain] (s1) [label=above:7] {\Large \$};

\node[snode,on chain] (s2) [label=above:8] {\Large a};

\node[tnode,on chain,red] (s5) [label=above:9] {\Large b};

\node[snode,on chain] (s6) [label=above:10] {\Large b};
\end{scope}

\begin{scope}[start chain=going right,node distance=7mm]

\node[snode,on chain] (f1) [label=below:1] {\Large \$};

\node[snode,on chain] (f2) [label=below:2] {\Large a};

\node[tnode,on chain,red] (f3) [label=below:3] {\Large b};

\node[snode,on chain] (f4) [label=below:4] {\Large b};

\node[tnode,on chain,red] (f5) [label=below:5] {\Large b};

\node[snode,on chain] (f6) [label=below:6] {\Large b};

\node[snode,on chain] (f7) [label=below:7] {\Large c};

\node[snode,on chain] (f8) [label=below:8] {\Large c};

\node[tnode,on chain,red] (f9) [label=below:9] {\Large d};

\node[snode,on chain] (f10) [label=below:10] {\Large e};
\end{scope}

    \draw[gray] (s1) -- (f7);
    \draw[gray] (s7) -- (f1);
    \draw[gray] (s2) -- (f8);
    \draw[gray] (s8) -- (f4);
    \draw[gray] (s4) -- (f6);
    \draw[gray] (s6) -- (f10);
    \draw[gray] (s10) -- (f2);
    \draw[red, line width=2.3pt] (s3) -- (f5);
    \draw[red, line width=2.3pt] (s5) -- (f9);
    \draw[red, line width=2.3pt] (s9) -- (f3);

\draw[gray, dashed] (s1) -- (f1);
\draw[gray] (s7) -- (f7);
\draw[gray] (s2) -- (f2);
\draw[gray] (s8) -- (f8);
\draw[gray] (s4) -- (f4);
\draw[gray] (s6) -- (f6);
\draw[gray] (s10) -- (f10);
\draw[red, line width=2.3pt] (s3) -- (f3);
\draw[red, line width=2.3pt] (s5) -- (f5);
\draw[red, line width=2.3pt] (s9) -- (f9);
\end{tikzpicture}}
}
\end{figure}


\begin{definition} \label{def:pseudo-cycle}
	Given a permutation $\pi$ of $\{1,\ldots, n\}$, a {\em pseudo-cycle} w.r.t.\ $\pi$ is a non-empty subset $S \subseteq \{1,\ldots,n\}$
	which can be partitioned into two subsets $S_{\links}$ and $S_{\rechts}$, possibly empty, such that $S_{\links} < S_{\rechts}$, and $\pi(S) = (S_{\links}-1) \cup S_{\rechts}$.
	Let $a = \max S_{\links}$, and $a = 0$ if $S_{\links}$ is empty. Further, let $b = \min S_{\rechts}$, and $b = n+1$ if $S_{\rechts}$ is empty. The {\em critical interval} $R \subseteq \{1,2,\ldots,n+1\}$ of the pseudo-cycle $S$ is defined as $R = [a+1,b]$. 
\end{definition}

For example, in Fig.~\ref{fig:cedcbbabb-a}, $S=\{3,5,8\}$ is a pseudo-cycle, with $S_{\links} = \{3,5\},$ $S_{\rechts} = \{8\}$, and $R = \{6,7,8\}$.  Note that every cycle $C$ of a permutation is a pseudo-cycle, with $C = C_{\rechts}$. In Fig.~\ref{fig:constraints15}, we highlighted two pseudo-cycles: $S_1 = \{6\}$, with critical interval $R_1 = \{7,8,9,10\}$, and $S_2 = \{7\}$, with $R_2 = \{1,2,3,4,5,6,7\}$. The elements of the critical interval are exactly those positions $i$ which turn $S$ into one or more cycles, when the $\$$ is inserted in position $i$ (see in Lemma~\ref{lemma:psudocycle-cycle}). In particular, with $S = S_{\rechts}=\{1,\ldots, n\}$, we get that $i=1$ is never nice. This is easy to see since, for every word $w$,  $1$ is a fixpoint in the standard permutation $\sigma_1$ of $\$w$.

\begin{definition}\label{def:shift}
	Let $1\leq i \leq  n+1$ and $S \subseteq \{1,2, \ldots , n\}$. We define 
	\begin{align*}
	\shift(S, i) &= \{ x \mid x \in S \text{ and } x < i \} \cup \{ x+1 \mid x \in S \text{ and } x \geq i \}, \; \text{and} \\
	\unshift(S, i) &= \{ x \mid x \in S \text{ and } x < i \} \cup \{ x-1 \mid x \in S \text{ and } x > i\}.
	\end{align*}
\end{definition}

\begin{lemma}\label{lemma:psudocycle-cycle}
	Let $w\in \Sigma^*$ and $\sigma = \sigma_w$. Let $1 \leq i \leq n+1$, and $U \subseteq \{1,2, \ldots , n+1\} \setminus \{i\}$. Then $U$ is a cycle in the permutation $\sigma_i$ if and only if $S=\unshift(U, i)$ is a pseudo-cycle w.r.t.\ $\sigma$, and $i$ belongs to the critical interval of $S$.
\end{lemma}

\begin{proof}
	Let  $U_{1} = \{ x \in U \mid x< i\}$, $U_{2} = \{ x \in U \mid x> i\}$. Then $S = U_{1} \cup (U_{2} -1)$. We have to show that $U$ is a cycle if and only if $S$ is a pseudo-cycle, with $S_{\links} = U_1$ and $S_{\rechts} = U_2-1$. Note that this implies that $i$ is contained in the critical interval of $S$. 
	
	First let $S$ be a pseudo-cycle with $S_{\links} = U_1$ and $S_{\rechts} = U_2-1$, and let $x\in U$. We have to show that $x\in \sigma_i(U)$, which implies the claim. If $x\in U_1$, then $x\in S_{\links}$, and there is a $y\in S$ s.t.\ $\sigma(y) = x-1$. If $y\in S_{\links}$, then $y\in U_1$ and $\sigma_i(y) = x$ by Lemma~\ref{lemma:sigma2sigmai}, thus $x\in \sigma_i(U)$.  Else $y\in S_{\rechts}$, then $y+1 \in U_2$ and $\sigma_i(y+1) = x$, again by Lemma~\ref{lemma:sigma2sigmai}, and thus $x\in \sigma_i(U)$. 
	
	Now let $x\in U_2$. Then $x-1\in S_{\rechts}$ and there is a $y\in S$ s.t.\ $\sigma(y)=x-1$. If $y\in S_{\links}$, then $y\in U_1$ and $x=\sigma(y) +1=\sigma_i(y)$ by Lemma~\ref{lemma:sigma2sigmai}, thus $x\in \sigma_i(U)$. Else $y\in S_{\rechts}$, then $y+1 \in U_2$ and $\sigma_i(y+1) = x$, again by Lemma~\ref{lemma:sigma2sigmai}, and thus $x\in \sigma_i(U)$.
	
	\medskip
	
	Conversely, let $U$ be a cycle, set $S_{\links} = U_1$ and $S_{\rechts} = U_2 - 1$. Let $x\in S$. We will show that if $x\in S_{\links}$, then $x-1\in \sigma(S)$, and if $x\in S_{\rechts}$, then $x\in \sigma(S)$, proving that $S$ is a pseudo-cycle. 
	The claim follows with analogous arguments as above and noting that $\sigma(j) = \sigma_i(j)-1$  if $j<i$, and  $\sigma_i(j+1)-1$ if $j\geq i$.

\end{proof}

\begin{theorem} \label{thm:pseudocycle}
	Let $w$ be a word of length $n$ over $\Sigma$, and $1\leq i \leq n+1$. Then $i$ is nice if and only if there is no pseudo-cycle $S$ w.r.t.\ the standard permutation $\sigma = \sigma_w$ whose critical interval contains $i$. 
\end{theorem} 

\begin{proof}

	``$\Rightarrow$'': We will assume that $\sigma$ contains a pseudo-cycle with $i$ in its critical interval, and show that then $i$ is not nice. 
	Let $S$ be a pseudo-cycle w.r.t.\ $\sigma$, $R$ its critical interval and $i \in R$. By Lemma~\ref{lemma:psudocycle-cycle}, $\shift(S,i)$ is a cycle in $\sigma_i$ not containing $i$. Therefore, $\sigma_i$ has at least two cycles, implying that $\dol(w, i) \not \in \BWT(\SigmaD)$.
	
	``$\Leftarrow$'': For the converse, assume that $i$ is not nice. Then $\sigma_i$ contains at least two cycles, and thus it contains at least one cycle $C\subseteq \{1,2,\ldots, n+1\}\setminus \{i\}$. By Lemma~\ref{lemma:psudocycle-cycle}, this implies that $\unshift(C,i)$ is a pseudo-cycle in $\sigma$, and its critical interval contains $i$. 
\end{proof}

With Theorem~\ref{thm:pseudocycle}, we can now prove the statements about our first example strings {\tt banana} and {\tt annnaa}. The word {\tt banana} has the pseudo-cycles $S_1 = \{2\}$ with critical interval $R_1=\{3,4,5,6,7\}$; and $S_2 = \{3,5,6\}$ with $R_2 = \{1,2,3\}$. Therefore, every position is contained in some critical interval. 
For the word {\tt annnaa}, we have $S_1=\{1\}$ with critical interval $R_1 = \{1\}$; $S_2 = \{2,3,4,5,6\}$ with $R_2 = \{1,2\}$; $S_3 = \{ 3,5\}$ with $R_3 = \{4,5\}$; $S_4 = \{4,6\}$ with $R_4 = \{5,6\}$; and all other pseudo-cycles are unions of these. The two positions $3$ and $7$ are not contained in any critical interval, and are therefore nice. In fact, {\tt an\$nnaa} = \BWT({\tt ananna\$}), and {\tt annnaa\$} = \BWT({\tt nanana\$}).

%

\section{Bounds on nice positions}\label{sec:parity}

In this section, we study the number of nice positions of a given string $w$. Recall that $\sigma_i$ is the standard permutation of $\dol(w,i)$. 

\begin{definition}\label{def:number_nice}
	For $w \in \Sigma^n$, let $h(w)$ denote the number of nice positions of $w$.
\end{definition}

We will first show that all nice positions of a word $w$ have the same parity. 

\begin{theorem}\label{thm:parity}
	Let $w$ be a word over $\Sigma$. Then either all nice positions are even, or all nice positions are odd. In particular, let $c$ be the number of cycles in the standard permutation $\sigma_w$; if $c$ is even, then all nice positions are odd, and if $c$ is odd, then all nice positions are even. 
\end{theorem}

\begin{proof}
	Let us assume that $i<j$ are both nice, thus $\sigma_i$ and $\sigma_j$ are cyclic. This implies that $\sgn(\sigma_i) = (-1)^n = \sgn(\sigma_j)$, since both cycles consist of $n+1$ elements. By Lemma~\ref{lemma:sigmai}, $\sigma_{i+1} = \tau_i \cdot \sigma_i$, where $\tau_i = (1, \sigma_i(i+1))$. Thus $\sigma_j = \tau_{j-1}\cdots \tau_i\cdot \sigma_i$, so $\sgn(\sigma_j) = (-1)^{j-i} \sgn(\sigma_i)$, and therefore $\sgn(\sigma_j) = \sgn(\sigma_i)$ if and only if $j-i$ is even. 
	
	Given a cycle $C=(x_1, \ldots, x_m)$, let $C' = (x_1+1, \ldots, x_m+1)$. Now let $\sigma_w = \prod_{j=1}^c C_j$ be the cycle decomposition of $\sigma_w$. By Lemma~\ref{lemma:sigmai}, $\sigma_1 = (1)\prod_{j=1}^c C'_j$. Therefore, $\sgn(\sigma_1) = (-1)^{n+1 - (c+1)} = (-1)^{n-c}$. On the other hand, again by Lemma~\ref{lemma:sigmai}, $\sigma_i = \tau_{i-1}\cdots \tau_1\cdot \sigma_1$, thus $\sgn(\sigma_i) = (-1)^{n-c+i-1}$. But this equals $(-1)^{n}$ if and only if $c$ and $i$ have different parity. 
\end{proof}

\begin{corollary}\label{coro:maxnice}
	Let $w\in \Sigma^n$. Then $h(w) \leq \lfloor \frac{n+1}{2} \rfloor$. 
\end{corollary}

\begin{proof} Follows from  Thm.~\ref{thm:parity} and the fact that $\sigma_1$ is not nice. 
\end{proof}

Given the cycle decomposition of $\sigma_w = \prod_{j=1}^c C_j$, let $\ell_j$ denote the minimum element of $C_j$, and $L = \max_{j=1,\ldots, c} \ell_j$.

\begin{proposition}\label{prop:L}
	If $i$ is nice, then $i\geq L+1$. In particular, $i\geq c+1$, where $c$ is the number of cycles of $\sigma_w$. 
\end{proposition}

\begin{proof}
	Note that every cycle $C_j$ is a pseudo-cycle, where $S_{\links} = \emptyset$ and $S_{\rechts}=C_j$, with critical interval $[1,\ell_j]$. Therefore, by Thm.~\ref{thm:pseudocycle}, no $i \leq L$ can be nice. The second claim follows since $L\geq \ell_c \geq c$. 
\end{proof}

\begin{corollary}
	Let $w\in \Sigma^n$. Then $h(w) \leq \lceil \frac{n-L+1}{2} \rceil$. 
\end{corollary}

\begin{proof}
	Follows from Thm.~\ref{thm:parity} and Prop.~\ref{prop:L}. 
\end{proof}

We next derive some properties of nice positions from Algorithm~\ref{algo:findPos}. 

\begin{proposition}\label{prop:c_i}
	Let $c_i$ be the number of cycles of $\sigma_i$. If $j$ is nice and $j>i$, then $j \geq i+c_i-1$. 
\end{proposition}

\begin{proof}
	The permutation $\sigma_j$ is computed from $\sigma_i$ by $j-i$ iterations of the for-loop of Algorithm~\ref{algo:findPos} (lines 9-17), each of which either results in incrementing (split) or decrementing (merge) the number of cycles. Therefore, at least $c_i-1$ steps are needed to arrive at a cyclic permutation. (Since $c_1 = c+1$, this implies in particular that for every nice position $j$, $j\geq c+1$, as already seen in Prop.~\ref{prop:L}.)
\end{proof}

\begin{definition}\label{def:bad-cycle}
	Let $C$ be a cycle of $\sigma_i$. We call $C$ a {\em bad cycle w.r.t.\ $i$} if $i\notin [\min C, \max C]$. 
\end{definition}

\addtocounter{exa}{-1}
\begin{exa}[continued from p.~\pageref{ex:sigmaPermutations1}]
	The cycle $(5,9)$ is a bad cycle w.r.t.\ $4$, and $(3,8,5,4,9)$ is a bad cycle w.r.t.\ $10$. 
\end{exa}
\addtocounter{exa}{+1}

\begin{proposition}\label{prop:badcycles}
	If $C$ is a bad cycle w.r.t.\ $i$, then 
	
	\begin{itemize}
		\item if $i< \min C$, then no $j\leq i$ is nice, 
		\item if $i> \max C$, then no $j\geq i$ is nice. 
	\end{itemize}
	
\end{proposition}

\begin{proof}
	Let $i < \min C$. Then $i$ is not nice, since $\sigma_i$ has at least two cycles. Now let $j<i$, thus $\sigma_i = \tau_{i-1}\ldots \tau_j \cdot \sigma_j$, where $\tau_k = (1, \sigma_k(k+1))$. Since $j \leq i < \min C$, it follows that each $\tau_k$ is disjoint from $C$, and since $C$ is a cycle of $\sigma_i$, therefore it is also a cycle of $\sigma_j$. Since $[\min C, \max C] \neq \{1,\ldots, n+1\}$, this implies that $j$ is not nice. 
	
	Analogously, if $i> \max C$, this implies that all $\sigma_j$ for $j\geq i$ have $C$ as a cycle, implying that $j$ is not nice. 
\end{proof}

\begin{definition}
	Let $\sigma_w = \prod_{j=1}^c C_j$ be the cycle decomposition of $\sigma_w$ and $\ell_j = \min C_j$ for $j=1, \ldots, c$,  where the cycles are in increasing order w.r.t.\ their minima, i.e.\ $\ell_1< \ldots < \ell_c$. 
	We call a pair $(\ell_j, \ell_{j}+1)$ {\em bad pair} if $j<c$ and $\ell_{j}+1\in C_j$. 
\end{definition}

\begin{exa}
	Given the permutation $(1)(2,6,3,7,9,4)(5,8,10)$ with $3$ cycles, the pair $(2,3)$ in the second cycle is a bad pair, and it is the only bad pair.
\end{exa}

\begin{proposition}\label{prop:badpairs}
	Let $b$ be the number of bad pairs in $\sigma_w$. If $i$ is a nice position of $w$, then $i \geq 2b+c$.
\end{proposition}

\begin{proof}
	We will count the number of iterations of the for-loop (lines 9-17) of Algorithm~\ref{algo:findPos} before arriving at a cyclic permutation. Let us refer to the iterations as either merge- or split-steps. As we saw before, we need at least $c$ merge-steps, since $\sigma_1$ has $c+1$ cycles. It should be also clear  that every additional split-step will necessitate a further merge-step. Therefore it suffices to show that every bad pair results in a distinct split-step. 
	
	Let $(\ell_j, \ell_j+1)$ be a bad pair. By Lemma~\ref{lemma:sigma2sigmai}, $\sigma_1 = (1) \prod C_j'$, where $\min C_j' = \ell_j+1$. Therefore, $C_j'$ is a bad cycle w.r.t.\ $\ell_j$, and thus is present in all $\sigma_i$ for $i\leq \ell_j$. Since $\ell_j \notin C_j'$, step $\ell_j$ is a merge-step. Now $\ell_{j}+2$ is still in $C_j'$, so step $\ell_{j}+1$ is a split-step. 
\end{proof}

\addtocounter{exa}{-1}
\begin{exa}[continued]
	The permutation $(1)(2,6,3,7,9,4)(5,8,10)$ is the standard permutation of the word {\tt abbababbaa}. There are two nice positions, $6$ and $8$, both greater or equal $5 = 2+3 = 2b+c$.
\end{exa}
\addtocounter{exa}{+1}

We summarize: 

\begin{theorem}
	Let $w$ be a word over $\Sigma$ and $\sigma_w = \prod_{j=1}^c C_j$ the cycle decomposition of its standard permutation $\sigma_w$. 
	
	\begin{enumerate}\label{thm:bounds}
		\item If $i$ is nice, then $i \geq \max\{L+1, 2b+c\}$, where $L = \max_j \min C_j$, and $b$ is the number of bad pairs in $\sigma_w$. 
		\item Let $c_i$ be the number of cycles of $\sigma_i$. If $j$ is nice, then $j \geq i+c_i-1$. Moreover, if $\sigma_i$ has a bad cycle $C$ s.t.\ $i>\max C$, then no $j\geq i$ is nice. 
	\end{enumerate}
\end{theorem}

Part {\em 1.} of Theorem~\ref{thm:bounds} can be used for heuristics to speed up the algorithm: Compute $i_0=\max\{L+1, 2b+c\}$, and start the algorithm with $\sigma_{i_0}$ instead of $\sigma_1$. Since $L,c,b$ can be computed in linear time by once scanning $\sigma_w$, and $\sigma_{i_0}$ can be computed in linear time, the total running time of the algorithm is still $\Oh(n \log n)$ in the worst case, but could be often faster in practice. 

Part {\em 2.} cannot be immediately turned into an algorithm improvement, because our current implementation does not allow extracting minima and maxima from cycles.


\section{Experimental result}\label{sec:results}

We implemented Algorithm {\sc FindNicePositions} (Algorithm~\ref{algo:findPos}) in Java (version 11). 
In the following, we give some examples (Sec.~\ref{sec:examples}), followed by some statistics on the number of nice positions (Sec.~\ref{sec:stats}). 

\subsection{Examples}\label{sec:examples}

We list all words over $\Sigma = \{{\tt a,b}\}$ of length $2$, $3$, $4$, and $5$ (Tables~\ref{tab:sigmaSize2_all} to~\ref{tab:sigmaSize5_all}). For each word $w$ (first column), we give the lexicographically smallest $v$ such that $\BWT(v) = w$, if such a $v$ exists, and a dash otherwise (second column); the standard permutation $\sigma=\sigma_w$ (third column); and the number $h(w)$ of nice positions for $w$ (fourth column). For each word $w$ with $h(w)>0$, we also list every $dol(w,i)$ with $i$ nice, giving the analogous information and specifying $i$ in the last (fifth) column.

In Table~\ref{tab:sigmaSize10-13-15-18_some}, the same information is shown about some longer strings over alphabets of size 2 and 3, ordered by string length ($n=10, 13, 15, 18$). We chose these strings in order to give examples of many different cases, as we detail next.

Among the words listed in Table~\ref{tab:sigmaSize10-13-15-18_some}, there are words which are BWTs of primitive words (strings 1, 7, 8, 12, 13, 20, 21, 22). Some of these have the maximum number of nice positions according to their length, recall Corollary~\ref{coro:maxnice} (strings 7, 12, 20); two words have only one nice position (8, 13); string 1 is an example that shows that, once a position is nice, not necessarily all following positions with the same parity are also nice. Similarly, string 21 shows that there are no further nice positions after position 12 due to a bad cycle with respect to 13 containing only elements strictly smaller than 13.

We further included BWTs of powers of a primitive words (strings 2, 14, 15, 23, 24, 25), but only one of these has the maximum possible number of nice positions (string 2). 

Finally, Table~\ref{tab:sigmaSize10-13-15-18_some} also contains strings that are not the BWT of any other word (strings 3, 4, 5, 6, 9, 10, 11, 16, 17, 18, 19). Three of these have no nice positions (strings 6, 11, 19). Sometimes the parity of the number of cycles equals the parity of the smallest element  of the last cycle (strings 3, 4, 5, 16, 18, 19), but this does not always happen (strings 6, 17).


\begin{table}
	\centering
	
	\bottomcaption{All strings $w$ of length 2 over a binary alphabet. See text for details. 
		\label{tab:sigmaSize2_all}}

	\begin{supertabular}{c @{\hspace{1cm}} c}
		\begin{tabular}{ | c | c | c | c | c |} \hline
			$w$ & $v$ & $\sigma$ & $h(w)$ & $i$ \\ 
			\hline \hline
			aa & $(\text{a})^2$ & $(1)(2)$ & $1$ & \\ 
			\hline 
			 
			aa\$ & aa\$ & $(1,2,3)$ & & 3 \\ 
			\hline \hline 
			
			ab & - & $(1)(2)$ & $1$ & \\ 
			\hline 
			
			ab\$ & ba\$ & $(1,2,3)$ & & 3 \\ 
			\hline 
		\end{tabular} & 
		
		\begin{tabular}{ | c | c | c | c | c |} \hline
			$w$ & $v$ & $\sigma$ & $h(w)$ & $i$ \\ 
			\hline \hline
			
			ba & ab & $(1,2)$ & $1$ &\\ 
			\hline 
			
			b\$a & ab\$ & $(1,2,3)$ & & 2 \\
			\hline \hline
			
			bb & $(\text{b})^2$ & $(1)(2)$ & $1$ & \\ 
			\hline 
			 
			bb\$ & bb\$ & $(1,2,3)$ & & 3 \\ 
			\hline 
			
		\end{tabular}
		\\
	\end{supertabular}
\end{table}


\begin{table}
	\centering
	
	\begin{tabular}{c @{\hspace{1cm}} c}
		\begin{tabular}{ | c | c | c | c | c |} \hline
			$w$ & $v$ & $\sigma$ & $h(w)$ & $i$ \\ 
			\hline \hline
			
			aaa & $(\text{a})^3$ & $(1)(2)(3)$ & $1$ & \\ 
			\hline 
			
			aaa\$ & aaa\$ & $(1,2,3,4)$ & & 4 \\ 
			\hline \hline 
			
			aab & - & $(1)(2)(3)$ & $1$ & \\ 
			\hline 
			
			aab\$ & baa\$ & $(1,2,3,4)$ & & 4 \\ 
			\hline \hline
			
			aba & - & $(1)(2,3)$ & $1$ &\\ 
			\hline 
			
			ab\$a & aba\$ & $(1,2,4,3)$ & & 3 \\
			\hline \hline
			
			abb & - & $(1)(2)(3)$ & $1$ & \\ 
			\hline 
			
			abb\$ & bba\$ & $(1,2,3,4)$ & & 4 \\ 
			\hline  
			
		\end{tabular} & 
		\begin{tabular}{ | c | c | c | c | c |} \hline
			$w$ & $v$ & $\sigma$ & $h(w)$ & $i$ \\ 
			\hline \hline
			baa & aab & $(1,3,2)$ & $1$ & \\ 
			\hline 
			
			b\$aa & aab\$ & $(1,4,3,2)$ & & 2 \\ 
			\hline \hline
			
			bab & - & $(1,2)(3)$ & $0$ &\\ 
			
			\hline \hline
			
			bba & abb & $(1,2,3)$ & $1$ &\\ 
			\hline 
			
			b\$ba & abb\$ & $(1,3,4,2)$ & & 2 \\ 
			bba\$ & bab\$ & $(1,3,2,4)$ & & 4 \\
			\hline \hline
			
			bbb & $(\text{b})^3$ & $(1)(2)(3)$ & $1$ & \\ 
			\hline 
			
			bbb\$ & bbb\$ & $(1,2,3,4)$ & & 4 \\ 
			\hline  
			
		\end{tabular}
		\\
		\end {tabular}
		
		\caption{All strings $w$ of length 3 over a binary alphabet. See text for details.
			\label{tab:sigmaSize3_all}}	
	\end{table}

	\begin{table}
		\centering
		
		\begin{tabular}{c @{\hspace{1cm}} c}
			\begin{tabular}{ | c | c | c | c | c |} \hline
				$w$ & $v$ & $\sigma$ & $h(w)$ & $i$ \\ 
				\hline \hline
				
				aaaa & $(\text{a})^4$ & $(1)(2)(3)(4)$ & $1$ & \\ 
				\hline 
				
				aaaa\$ & aaaa\$ & $(1,2,3,4,5)$ & & 5 \\ 
				\hline \hline 
				
				aaab & - & $(1)(2)(3)(4)$ & $1$ & \\ 
				\hline 
				
				aaab\$ & baaa\$ & $(1,2,3,4,5)$ & & 5 \\ 
				\hline \hline
				
				aaba & - & $(1)(2)(3,4)$ & $1$ &\\ 
				\hline 
				
				aab\$a & abaa\$ & $(1,2,3,5,4)$ & & 4 \\
				\hline \hline
				
				aabb & - & $(1)(2)(3)(4)$ & $1$ & \\ 
				\hline 
				
				aabb\$ & bbaa\$ & $(1,2,3,4,5)$ & & 5 \\ 
				\hline \hline
				
				abaa & - & $(1)(2,4,3)$ & $1$ & \\ 
				\hline 
				
				ab\$aa & aaba\$ & $(1,2,5,4,3)$ & & 3 \\ 
				\hline \hline
				
				abab & - & $(1)(2,3)(4)$ & $0$ & \\ 
				
				\hline \hline
				
				abba & - & $(1)(2,3,4)$ & $2$ & \\ 
				\hline 
				
				ab\$ba & abba\$ & $(1,2,4,5,3)$ & & 3 \\ 
				abba\$ & baba\$ & $(1,2,4,3,5)$ & & 5 \\ 
				\hline \hline
				
				abbb & - & $(1)(2)(3)(4)$ & $1$ & \\ 
				\hline 
				
				abbb\$ & bbba\$ & $(1,2,3,4,5)$ & & 5 \\ 
				\hline  
			\end{tabular} & 
			\begin{tabular}{ | c | c | c | c | c |} \hline
				$w$ & $v$ & $\sigma$ & $h(w)$ & $i$ \\ 
				\hline \hline
				baaa & aaab & $(1,4,3,2)$ & $1$ & \\ 
				\hline 
				
				b\$aaa & aaab\$ & $(1,5,4,3,2)$ & & 2 \\ 
				\hline \hline
				
				baab & - & $(1,3,2)(4)$ & $0$ &\\ 
				
				\hline \hline
				
				baba & aabb & $(1,3,4,2)$ & $1$ &\\ 
				\hline  
				
				b\$aba & aabb\$ & $(1,4,5,3,2)$ & & 2 \\ 
				\hline \hline
				
				babb & - & $(1,2)(3)(4)$ & $0$ & \\ 
				
				\hline \hline
				
				bbaa & $(\text{ab})^2$ & $(1,3)(2,4)$ & $2$ &\\ 
				\hline  
				
				bb\$aa & abab\$ & $(1,4,2,5,3)$ & & 3 \\ 
				bbaa\$ & baab\$ & $(1,4,3,2,5)$ & & 5 \\ 
				\hline \hline

				bbab & - & $(1,2,3)(4)$ & $1$ &\\ 
				\hline  
				
				bbab\$ & bbab\$ & $(1,3,2,4,5)$ & & 5 \\ 
				\hline \hline
				
				bbba & abbb & $(1,2,3,4)$ & $2$ &\\ 
				\hline  
				
				b\$bba & abbb\$ & $(1,3,4,5,2)$ & & 2 \\ 
				bbb\$a & babb\$ & $(1,3,5,2,4)$ & & 4 \\ 
				\hline \hline
				
				bbbb & $(\text{b})^4$ & $(1)(2)(3)(4)$ & $1$ & \\ 
				\hline  
				
				bbbb\$ & bbbb\$ & $(1,2,3,4,5)$ & & 5 \\ 
				\hline  
				
			\end{tabular}
			\\ 
		\end{tabular}
		\caption{All strings $w$ of length 4 over a binary alphabet. See text for details.
			\label{tab:sigmaSize4_all}}	
	\end{table}
	
	
	\begin{table}
		\centering
		
		\begin{tabular}{c @{\hspace{1cm}} c}
			\begin{tabular}{ | c | c | c | c | c |} \hline
				$w$ & $v$ & $\sigma$ & $h(w)$ & $i$ \\ 
				\hline \hline
				
				aaaaa & $(\text{a})^5$ & $(1)(2)(3)(4)(5)$ & $1$ & \\ 
				\hline 
				
				aaaaa\$ & aaaaa\$ & $(1,2,3,4,5,6)$ & & 6 \\ 
				\hline \hline 
				
				aaaab & - & $(1)(2)(3)(4)(5)$ & $1$ & \\ 
				\hline 
				
				aaaab\$ & baaaa\$ & $(1,2,3,4,5,6)$ & & 6 \\ 
				\hline \hline
				
				aaaba & - & $(1)(2)(3)(4,5)$ & $1$ &\\ 
				\hline 
				
				aaab\$a & abaaa\$ & $(1,2,3,4,6,5)$ & & 5 \\
				\hline \hline
				
				aaabb & - & $(1)(2)(3)(4)(5)$ & $1$ & \\ 
				\hline 
				
				aaabb\$ & bbaaa\$ & $(1,2,3,4,5,6)$ & & 6 \\ 
				\hline \hline
				
				aabaa & - & $(1)(2)(3,5,4)$ & $1$ & \\ 
				\hline 
				
				aab\$aa & aabaa\$ & $(1,2,3,6,5,4)$ & & 4 \\ 
				\hline \hline
				
				aabab & - & $(1)(2)(3,4)(5)$ & $0$ & \\ 
				
				\hline \hline
				
				aabba & - & $(1)(2)(3,4,5)$ & $2$ & \\ 
				\hline  
				
				aab\$ba & abbaa\$ & $(1,2,3,5,6,4)$ & & 4 \\ 
				aabba\$ & babaa\$ & $(1,2,3,5,4,6)$ & & 6 \\ 
				\hline \hline
				
				aabbb & - & $(1)(2)(3)(4)(5)$ & $1$ & \\ 
				\hline  
				
				aabbb\$ & bbbaa\$ & $(1,2,3,4,5,6)$ & & 6 \\ 
				\hline \hline
				
				abaaa & - & $(1)(2,5,4,3)$ & $1$ & \\ 
				\hline  
				
				ab\$aaa & aaaba\$ & $(1,2,6,5,4,3)$ & & 3 \\ 
				\hline \hline
				
				abaab & - & $(1)(2,4,3)(5)$ & $0$ &\\ 
				
				\hline \hline
				
				ababa & - & $(1)(2,4,5,3)$ & $1$ &\\ 
				\hline  
				
				ab\$aba & aabba\$ & $(1,2,5,6,4,3)$ & & 3 \\ 
				\hline \hline
				
				ababb & - & $(1)(2,3)(4)(5)$ & $0$ & \\ 
				
				\hline \hline
				
				abbaa & - & $(1)(2,4)(3,5)$ & $2$ &\\ 
				\hline 
				
				abb\$aa & ababa\$ & $(1,2,5,3,6,4)$ & & 4 \\ 
				abbaa\$ & baaba\$ & $(1,2,5,4,3,6)$ & & 6 \\ 
				\hline \hline

				abbab & - & $(1)(2,3,4)(5)$ & $1$ &\\ 
				\hline  
				
				abbab\$ & bbaba\$ & $(1,2,4,3,5,6)$ & & 6 \\ 
				\hline \hline
				
				abbba & - & $(1)(2,3,4,5)$ & $2$ &\\ 
				\hline 
				
				ab\$bba & abbba\$ & $(1,2,4,5,6,3)$ & & 3 \\ 
				abbb\$a & babba\$ & $(1,2,4,6,3,5)$ & & 5 \\ 
				\hline \hline
				
				abbbb & - & $(1)(2)(3)(4)(5)$ & $1$ & \\ 
				\hline 
				
				abbbb\$ & bbbba\$ & $(1,2,3,4,5,6)$ & & 6 \\ 
				\hline 
				
			\end{tabular} & 
			\begin{tabular}{ | c | c | c | c | c |} \hline
				$w$ & $v$ & $\sigma$ & $h(w)$ & $i$ \\ 
				\hline \hline
				baaaa & aaaab & $(1,5,4,3,2)$ & $1$ & \\ 
				\hline  
				
				b\$aaaa & aaaab\$ & $(1,6,5,4,3,2)$ & & 2 \\ 
				\hline \hline 
				
				baaab & - & $(1,4,3,2)(5)$ & $0$ & \\ 
				
				\hline \hline
				
				baaba & aaabb & $(1,4,5,3,2)$ & $1$ &\\ 
				\hline  
				
				b\$aaaba & aaabb\$ & $(1,5,6,4,3,2)$ & & 2 \\
				\hline \hline
				
				baabb & - & $(1,3,2)(4)(5)$ & $0$ & \\ 
				
				\hline \hline
				
				babaa & - & $(1,4,2)(3,5)$ & $0$ & \\ 
				
				\hline \hline
				
				babab & - & $(1,3,4,2)(5)$ & $0$ & \\ 
				
				\hline \hline
				babba & aabbb & $(1,3,4,5,3)$ & $1$ & \\ 
				\hline  
				
				b\$abba & aabbb\$ & $(1,4,5,6,3,2)$ & & 2 \\ 
				\hline \hline
				
				babbb & - & $(1,2)(3)(4)(5)$ & $0$ & \\ 
				
				\hline \hline
				bbaaa & aabab & $(1,4,2,5,3)$ & $3$ & \\ 
				\hline  
				
				b\$baaa & aabab\$ & $(1,5,3,6,4,2)$ & & 2 \\ 
				bba\$aa & abaab\$ & $(1,5,3,2,6,4)$ & & 4 \\ 
				bbaaa\$ & baaab\$ & $(1,5,4,3,2,6)$ & & 6 \\ 
				\hline \hline
				
				bbaab & - & $(1,3)(2,4)(5)$ & $1$ &\\ 
				\hline  
				
				bbaab\$ & bbaab\$ & $(1,4,3,2,5,6)$ & & 6 \\ 
				\hline \hline
				
				bbaba & - & $(1,3)(2,4,5)$ & $2$ &\\ 
				\hline  
				
				bb\$aba & abbab\$ & $(1,4,2,5,6,3)$ & & 3 \\ 
				bbab\$a & baabbb\$ & $(1,4,6,3,2,5)$ & & 5 \\ 
				\hline \hline
				
				bbabb & - & $(1,2,3)(4)(5)$ & $1$ & \\ 
				\hline  
				
				bbabb\$ & bbbab\$ & $(1,3,2,4,5,6)$ & & 6 \\ 
				\hline \hline
				
				bbbaa & ababb & $(1,3,5,2,4)$ & $2$ &\\ 
				\hline  
				
				b\$bbaa & ababb\$ & $(1,4,6,3,5,2)$ & & 2 \\ 
				bbbaa\$ & babab\$ & $(1,4,2,5,3,6)$ & & 6 \\ 
				\hline \hline
				
				bbbab & - & $(1,2,3,4)(5)$ & $0$ &\\ 
				
				\hline \hline
				
				bbbba & abbbb & $(1,2,3,4,5)$ & $3$ &\\ 
				\hline  
				
				b\$bbba & abbbb\$ & $(1,3,4,5,6,2)$ & & 2 \\ 
				bbb\$ba & babbb\$ & $(1,3,5,6,2,4)$ & & 4 \\ 
				bbbba\$ & bbabb\$ & $(1,3,5,2,4,6)$ & & 6 \\ 
				\hline \hline
				
				bbbbb & $(\text{b})^5$ & $(1)(2)(3)(4)(5)$ & $1$ & \\ 
				\hline  
				
				bbbbb\$ & bbbbb\$ & $(1,2,3,4,5,6)$ & & 6 \\ 
				\hline  
				
			\end{tabular}
			\\ \\
		\end{tabular} 
		\caption{All strings $w$ of length 5 over a binary alphabet. See text for details.
			\label{tab:sigmaSize5_all}}	
	\end{table}
	

\begin{landscape}

\centering

\begingroup
\setlength{\LTleft}{-20cm plus -1fill}
\setlength{\LTright}{\LTleft}
\begin{longtable}{ | r | c | c | c | c | c | c |} \hline
	& $w$ & $v$ & $\sigma$ & $c$ & nice positions & $h(w)$ \\ 
	\hline
  	\endhead

  	1 & bcacbaaacb & aabccacabb & $(1,5,6,2,8,4,9,10,7,3)$ & 1 & $2,4,8,10$ & 4 \\ \hline
  	2 & ccaaaaaabb & $(\text{aaabc})^2$ & $(1,9,7,5,3)(2,10,8,6,4)$ & 2 & $3,5,7,9,11$ & 5 \\ \hline
  	3 & cbcaaacaba & - & $(1, 8, 4)(2, 6, 3, 9, 7, 10, 5)$ & 2 & $3,7,11$ & 3 \\ \hline
  	4 & accbcbaaaa & - & $(1)(2,8,3,9,4,6,7)(5,10)$ & 3 & $6,8,10$ & 3 \\ \hline 
  5 & ccaaabcaac & - & $(1,7,9,5,3)(2,8,4)(6)(10)$ & 4 & $11$ & 1 \\ \hline
  	6 & bacbacacab & - & $(1,5,2)(3,8,10,7)(4,6,9)$ & 3 & - & 0 \\ \hline \hline

  	7 & bbaabbbbbbbba & aabbbbbbbbbab & $(1,4,2,5,6,7,8,9,10,11,12,13,3)$ & 1 & 2,4,6,8,10,12,14 & 7 \\ \hline
  	   8 & babbbbabbbbba & aabbabbbbbbbb & $(1,4,6,8,9,10,11,12,13,3,5,7,2)$ & 1 & 2 & 1 \\ \hline
  	  9 & abbabbbbbabba & - & $(1)(2,5,7,9,11,12,13,4)(3,6,8,10)$ & 3 & 4,8,10 & 3 \\ \hline
  	   10 & abbaaaaaaaaaa & - & $(1)(2,12,10,8,6,4)(3,13,11,9,7,5)$ & 3 & 4,6,8,10,12,14 & 6 \\ \hline
  	   11 & babbaaabaaaba & - & $(1,9,5,2)(3,10,6)(4,11,7)(8,12,13)$ & 4 & - & 0 \\ \hline \hline

  	12 & bbaaaabbbbbbbba & aaabbbbbbbbbaab & $(1,6,4,2,7,8,9,10,11,12,13,14,15,5,3)$ & 1 & 2,4,6,8,10,12,14,16 & 8 \\ \hline
  	   13 & babababababbbba & aaabaabbbbbbabb & $(1,7,10,5,9,11,12,13,14,15,6,3,8,4,2)$ & 1 & 2 & 1 \\ \hline
  	  14 & bbbaaaaaaaaaaaa & (aaaab)$^3$ & $(1,13,10,7,4)(2,14,11,8,5)(3,15,12,9,6)$ & 3 & 4,6,10,12,16 & 5 \\ \hline
  	  15  & bbbbbaaaaaaaaaa & (aab)$^5$ & $(1,11,6)(2,12,7)(3,13,8)(4,14,9)(5,15,10)$ & 5 & 6,8,10,14,16 & 5 \\ \hline
  	   16 & bbaababaaabbbab & - & $(1,8,4,2,9,5,10,6,3)(7,11,12,13,14)(15)$ & 3 & 16 & 1 \\ \hline
  	  17  & bbaababaaabbbaa & - & $(1,9,5,11,13,15,8,4,2,10,6,3)(7,12,14)$ & 2 & 9,11,15 & 3 \\ \hline
  	  18  & bbabaabaaabbaaa & - & $(1,10,6,3)(2,11,14,8,4,12,15,9,5)(7,13)$ & 3 & 10,12,14,16 & 4 \\ \hline
  	   19 & babbabbababaaab & - & $(1,8,3,9,13,6,11,14,7,12,5,2)(4,10)(15)$ & 3 & - & 0 \\ \hline \hline

  	20 & bbaababbbabbabaaaa & aaababbabbababbaab & $(1,10,4,2,11,16,7,13,5,12,17,8,14,18,9,15,6,3)$ & 1 & 2,4,6,8,10,12,14,16,18 & 9 \\ \hline
  	 21 & bbaaaaaaabbabbbbba & aaaaabbbbbbbaaaabb & $(1,10,12,8,6,4,2,11,13,14,15,16,17,18,9,7,5,3)$ & 1 & 2,4,6,8,10,12 & 6 \\ \hline
  	  22 & bbbaaabaaabaaababa & aaabaaabaaababbbab & $(1,12,7,15,17,18,11,16,10,6,3,14,9,5,2,13,8,4)$ & 1 & 2,6,18 & 3 \\ \hline
  	  23 & bbaaaabbbbaabbbbaa & (aaababbbb)$^2$ & $(1,9,13,15,17,7,11,5,3)(2,10,14,16,18,8,12,6,4)$ & 2 & 3,5,7,9,11 & 5 \\ \hline
  	  24 & bbbaaabbbbbbaaaaaa & (aababb)$^3$ & $(1,10,16,7,13,4)(2,11,17,8,14,5)(3,12,18,9,15,6)$ & 3 & 4,6,10,12 & 4 \\ \hline
  	  25 & bbbbbbbbbbbbaaaaaa & (abb)$^6$ & $(1,7,13)(2,8,14)(3,9,15)(4,10,16)(5,11,17)(6,12,18)$ & 6 & 7,9,11,13,17,19 & 6 \\ \hline

\caption{Some examples of words over alphabets of size 2 and 3; of length 10, 13, 15 18. (See text for more details.) \label{tab:sigmaSize10-13-15-18_some}}
\end{longtable}
\endgroup

\end{landscape}



\subsection{Statistics}\label{sec:stats}

In Tables~\ref{tab:tablePercentages2-short} and~\ref{tab:tablePercentages3-short}, we present statistics on the number of nice positions $h(w)$. Table~\ref{tab:tablePercentages2-short} contains the statistics for a binary alphabet and $n=19, 20$. We give the statistics for all $n=3, \ldots, 20$ in the Appendix (Table~\labelcref{tab:tablePercentages2}). Analogously, Table~\ref{tab:tablePercentages3-short} contains the same information for a ternary alphabet and $n=19, 20$, while, in the Appendix, statistics for all $n=3, \ldots, 20$ are shown (Table~\ref{tab:tablePercentages3}).

For fixed $n$, we give the absolute number of strings of length $n$ with $k$ nice positions (column 3), as well as the corresponding percentage (column 4). Percentages have been rounded to the next integer, where we write `$<0.5$' for percentages $x$ such that $0<x<0.5\%$. 
In columns 5 and 6, we give those strings with $k$ nice positions which are not BWT images, in absolute and percentage numbers; in columns 7 and 8, the same for BWT images. The last two columns contain a subdivision of column 7: the number of BWT images with $k$ nice positions which are BWTs of a primitive word (column 9) and of powers of primitive words (column 10).

\begin{center}
	

\begingroup
\setlength{\LTleft}{-20cm plus -1fill}
\setlength{\LTright}{\LTleft}

\bottomcaption{Statistics of words of length 19 and 20 over a binary alphabet. For the full table for $n=3,\ldots,20$, see the Appendix. Percentages rounded to nearest integer. See text for further details. 
	\label{tab:tablePercentages2-short}}
\tablehead{\hline ${\bf |\Sigma| = 2}$ && && && &&\multicolumn{2}{c|}{{\small BWTs of}} \\
	& $h(w)$ & {\small all} & \% & {\small noBWTs} & \% & {\small BWTs} & \% & {\small prim} & {\small pow} \\ \hline}

\begin{supertabular}{c || r | r | r | r | r | r | r || r | r |} \hline
	
$n=19$ & 0 & 335796 & 64 & 335796 & 68 & 0 & 0 & 0 & 0\\  
& 1 & 75616 & 14 & 60878 & 12 & 14738 & 53 & 14736 & 2\\  
& 2 & 22341 & 4 & 21347 & 4 & 994 & 4 & 994 & 0\\  
& 3 & 26820 & 5 & 24874 & 5 & 1946 & 7 & 1946 & 0\\  
& 4 & 25222 & 5 & 22776 & 5 & 2446 & 9 & 2446 & 0\\  
& 5 & 18916 & 4 & 16298 & 3 & 2618 & 9 & 2618 & 0\\  
& 6 & 11791 & 2 & 9520 & 2 & 2271 & 8 & 2271 & 0\\  
& 7 & 5428 & 1 & 3950 & 1 & 1478 & 5 & 1478 & 0\\  
& 8 & 1932 & $<0.5$ & 1131 & $<0.5$ & 801 & 3 & 801 & 0\\  
& 9 & 396 & $<0.5$ & 122 & $<0.5$ & 274 & 1 & 274 & 0\\  
& 10 & 30 & $<0.5$ & 0 & 0 & 30 & $<0.5$ & 30 & 0\\  
\hline
& {\small total}  & 524288 & 100 & 496692 & 100 & 27596 & 100 & 27594 & 2\\  
\hline \hline
$n=20$ & 0 & 679618 & 65 & 679618 & 68 & 0 & 0 & 0 & 0\\  
& 1 & 144564 & 14 & 116554 & 12 & 28010 & 53 & 28008 & 2\\  
& 2 & 44031 & 4 & 42048 & 4 & 1983 & 4 & 1983 & 0\\  
& 3 & 52254 & 5 & 48274 & 5 & 3980 & 8 & 3971 & 9\\  
& 4 & 48598 & 5 & 43639 & 4 & 4959 & 9 & 4950 & 9\\  
& 5 & 37512 & 4 & 32411 & 3 & 5101 & 10 & 5084 & 17\\  
& 6 & 24195 & 2 & 20011 & 2 & 4184 & 8 & 4159 & 25\\  
& 7 & 11780 & 1 & 9185 & 1 & 2595 & 5 & 2564 & 31\\  
& 8 & 4718 & $<0.5$ & 3471 & $<0.5$ & 1247 & 2 & 1234 & 13\\  
& 9 & 1198 & $<0.5$ & 809 & $<0.5$ & 389 & 1 & 385 & 4\\  
& 10 & 108 & $<0.5$ & 68 & $<0.5$ & 40 & $<0.5$ & 39 & 1\\  
\hline
& {\small total}  & 1048576 & 100 & 996088 & 100 & 52488 & 100 & 52377 & 111\\ 
\hline \hline

\end{supertabular}
\endgroup

	
\begingroup
\setlength{\LTleft}{-20cm plus -1fill}
\setlength{\LTright}{\LTleft}

\bottomcaption{Statistics of words of length 19 and 20 over a ternary alphabet. For the full table for $n=3,\ldots,20$, see the Appendix. Percentages rounded to nearest integer. See text for further details. 
	\label{tab:tablePercentages3-short}}
\tablehead{\hline ${\bf |\Sigma| = 3}$ && && && &&\multicolumn{2}{c|}{{\small BWTs of}} \\
	& $h(w)$ & {\small all} & \% & {\small noBWTs} & \% & {\small BWTs} & \% & {\small prim} & {\small pow} \\ \hline}

\begin{supertabular}{c || r | r | r | r | r | r | r || r | r |} \hline
	
$n=19$ & 0 & 705882210 & 61 & 705882210 & 64 & 0 & 0 & 0 & 0\\  
& 1 & 167901165 & 14 & 143474476 & 13 & 24426689 & 40 & 24426686 & 3\\  
& 2 & 93739960 & 8 & 86429556 & 8 & 7310404 & 12 & 7310404 & 0\\  
& 3 & 78834558 & 7 & 70340339 & 6 & 8494219 & 13 & 8494219 & 0\\  
& 4 & 54399352 & 5 & 46833919 & 4 & 7565433 & 12 & 7565433 & 0\\  
& 5 & 32418938 & 3 & 26772240 & 2 & 5646698 & 9 & 5646698 & 0\\  
& 6 & 17384924 & 2 & 13585965 & 1 & 3798959 & 6 & 3798959 & 0\\  
& 7 & 8057214 & 1 & 5762779 & 1 & 2294435 & 4 & 2294435 & 0\\  
& 8 & 2889010 & $<0.5$ & 1743235 & $<0.5$ & 1145775 & 2 & 1145775 & 0\\  
& 9 & 678872 & $<0.5$ & 265089 & $<0.5$ & 413783 & 1 & 413783 & 0\\  
& 10 & 75264 & $<0.5$ & 0 & 0 & 75264 & $<0.5$ & 75264 & 0\\  
\hline
& {\small total} & 1162261467 & 100 & 1101089808 & 100 & 61171659 & 100 & 61171656 & 3\\
\hline
\hline	
$n=20$ & 0 & 2142855232 & 61 & 2142855232 & 65 & 0 & 0 & 0 & 0\\	
&1 & 485799027 & 14 & 416406087 & 13 & 69392940 & 40 & 69392937 & 3	
\\
&2 & 272773366 & 8 & 251795904 & 8 & 20977462 & 12 & 20977186 & 276 \\
&3 & 230944816 & 7 & 206174661 & 6 & 24770155 & 14 & 24769338 & 817 \\
&4 & 161668972 & 5 & 139645685 & 4 & 22023287 & 13 & 22022351 & 936 \\
&5 & 98349528 & 3 & 82124392 & 3 & 16225136 & 9 & 16224109 & 1027 \\
&6 & 54140632 & 2 & 43539360 & 1 & 10601272 & 6 & 10600113 & 1159 \\
&7 & 26452876 & 1 & 20257612 & 1 & 6195264 & 4 & 6194176 & 1088 \\
&8 & 10452426 & $<0.5$ & 7480470 & $<0.5$ & 2971956 & 2 & 2971438 & 518 \\
&9 & 2927926 & $<0.5$ & 1915361 & $<0.5$ & 1012565 & 1 & 1012448 & 117 \\
&10 & 419600 & $<0.5$ & 247421 & $<0.5$ & 172179 & $<0.5$ & 172168 & 11 \\
\hline
& {\small total} & 3486784401 & 100 & 3312442185 & 100 & 174342216 & 100 & 174336264 & 5952\\
\hline \hline

\end{supertabular}
\endgroup

	
\end{center}

\normalsize


%
\section{Conclusion}\label{sec:conclusion}

In this paper, we studied a combinatorial question on the Burrows-Wheeler transform, namely in which positions (called {\em nice} positions) the sentinel character can be inserted in order to turn a given word $w$ into a BWT image. We developed a combinatorial characterization of nice positions and presented an efficient algorithm to compute all nice  positions of the word. We also showed that all nice positions have the same parity, and were able to give lower bounds on the values of nice positions, as well as upper bounds on the number of nice positions. These results are based  on properties of the standard permutation of the original word. 

We also included in the paper a number of examples for short strings over an alphabet of cardinality 2 resp.\ 3, as well as some statistics regarding the number of nice positions of a word.  

We pose the following open problems:
	\begin{enumerate}
		\item Find tight lower and upper bounds on $h(w)$, the number of nice positions of $w$. 
		\item Develop a $o(n \log n)$ time algorithm for computing all nice positions of a word of length $n$. One possibility to achieve this would be to find a data structure for maintaining the cycle structure of the standard permutation that allows constant-time updates. 
		Another direction towards a faster algorithm could be to design an output-sensitive algorithm, which takes advantage of the bounds developed in Section~\ref{sec:parity}. 
		\item Given $w$, compute $h(w)$ without explicitly computing  all nice positions, possibly by using the pseudo-cycle characterization of nice positions from Section~\ref{sec:characterization}. 
\end{enumerate}


\section*{Acknowledgements}

We thank the Scuola di Dottorato di Scienze Naturali e Ingegneristiche dell'Universit\`a degli Studi di Verona for funding SG. All authors thank the members of the Monday Meetings of the Algorithms Group at the Dept.\ of Computer Science, University of Verona, for interesting discussions.

\bibliographystyle{abbrv}
\bibliography{Whenadollar.bib}

\bigskip

\newpage

\section*{APPENDIX}

In this appendix, we give the algorithm for computing the standard permutation (Algorithm~\ref{algo:sigma}), two further examples for Algorithm~\ref{algo:findPos} (Examples~\ref{ex:sigmaPermutations2} and \ref{ex:sigmaPermutations3}), and the full splay tree implementation for Example~\ref{ex:sigmaPermutations1} (Fig.~\ref{fig:sigmaPermutations1}). Finally, we include the full versions of Tables~\ref{tab:tablePercentages2-short} and~\ref{tab:tablePercentages3-short}: Table~\ref{tab:tablePercentages2} contains statistics for $|\Sigma|=2$ and $n=3, \ldots, 20$, while Table~\ref{tab:tablePercentages3} contains statistics for the same lengths and $|\Sigma|=3$. 

\IncMargin{1em}
\begin{algorithm}
	\SetKw{DownTo}{down to}
	\DontPrintSemicolon
	{Given a word $w \in \Sigma^n$, compute its standard permutation $\sigma$.}\;
	\LinesNumbered
	\BlankLine
	
	\nl $n \gets |w|$\;
	\nl $count \gets$ array of length $|\Sigma|$ of zeros \tcp*[l]{\parbox[t]{1.5in}{\raggedright $count[j]$ is no.\ of occ's of $j$th character in $w$}}\nonl 
	\nl \For{$i \gets 1$ \KwTo $n$}{
		$j \gets$ lexicographic rank of $w_i$\;
		$count[j] \gets count[j]+1$\;
	}
	
	\For{$j=2$ \KwTo $|\Sigma|$}{
		$count[j] \gets count[j] + count[j-1]$\;
	}
	
	\For{$i=n$ \DownTo $1$}{
		$j \gets$ lexicographic rank of $w_i$\;
		$\sigma(i) \gets count[j]$\;
		$count[j] \gets count[j]-1$
	}
	
	\Return{$\sigma$}
	\caption{{\sc Standard Permutation}$(w)$\label{algo:sigma}}
	
\end{algorithm}
\DecMargin{1em}

\newpage

\begin{exa}\label{ex:sigmaPermutations2}

$$w = {\tt acbcccbcca} \qquad \sigma = \bigl( \begin{smallmatrix} 1 & 2 & 3 & 4 & 5 & 6 & 7 & 8 & 9 & 10 \\ 1 & 5 & 3 & 6 & 7 & 8 & 4 & 9 & 10 & 2 \end{smallmatrix} \bigr) = (1)(2, 5, 7, 4, 6, 8, 9, 10)(3)$$
\centering
\begin{minipage}[t]{0.8\textwidth}
	
	\vspace{0.5cm}
	
	$\sigma_1 = \bigl( \begin{smallmatrix} \fat{\red{1}} & 2 & 3 & 4 & 5 & 6 & 7 & 8 & 9 & 10 & 11 \\ \fat{\red{1}} & 2 & 6 & 4 & 7 & 8 & 9 & 5 & 10 & 11 & 3 \end{smallmatrix} \bigr) = (\fat{\red{1}})(\fat{\red{2}})(3, 6, 8, 5, 7, 9, 10, 11)(4)$
	\bigskip
	
	$\sigma_2 = \bigl( \begin{smallmatrix} 1 & 2 & 3 & 4 & 5 & 6 & 7 & 8 & 9 & 10 & 11 \\ \fat{\red{2}} & \fat{\red{1}} & 6 & 4 & 7 & 8 & 9 & 5 & 10 & 11 & 3 \end{smallmatrix} \bigr) = (1, \fat{\red{2}})(\fat{\red{3}}, 6, 8, 5, 7, 9, 10, 11)(4)$
	\bigskip
	
	$\sigma_3 = \bigl( \begin{smallmatrix} 1 & 2 & 3 & 4 & 5 & 6 & 7 & 8 & 9 & 10 & 11 \\ 2 & \fat{\red{6}} & \fat{\red{1}} & 4 & 7 & 8 & 9 & 5 & 10 & 11 & 3 \end{smallmatrix} \bigr) = (1, 2, 6, 8, 5, 7, 9, 10, 11, \fat{\red{3}})(\fat{\red{4}})$
	\bigskip
	
	\fbox{$\sigma_4 = \bigl( \begin{smallmatrix} 1 & 2 & 3 & 4 & 5 & 6 & 7 & 8 & 9 & 10 & 11 \\ 2 & 6 & \fat{\red{4}} & \fat{\red{1}} & 7 & 8 & 9 & 5 & 10 & 11 & 3 \end{smallmatrix} \bigr) = (1, 2, 6, 8, \fat{\red{5}}, 7, 9, 10, 11, 3 ,\fat{\red{4}})$}
	\bigskip
	
	$\sigma_5 = \bigl( \begin{smallmatrix} 1 & 2 & 3 & 4 & 5 & 6 & 7 & 8 & 9 & 10 & 11 \\ 2 & 6 & 4 & \fat{\red{7}} & \fat{\red{1}} & 8 & 9 & 5 & 10 & 11 & 3 \end{smallmatrix} \bigr) = (1, 2, \fat{\red{6}}, 8, \fat{\red{5}})(7, 9, 10, 11, 3 ,4)$
	\bigskip
	
	$\sigma_6 = \bigl( \begin{smallmatrix} 1 & 2 & 3 & 4 & 5 & 6 & 7 & 8 & 9 & 10 & 11 \\ 2 & 6 & 4 & 7 & \fat{\red{8}} & \fat{\red{1}} & 9 & 5 & 10 & 11 & 3 \end{smallmatrix} \bigr) = (1, 2, \fat{\red{6}})(8, 5)(\fat{\red{7}}, 9, 10, 11,3 ,4)$
	\bigskip
	
	$\sigma_7 = \bigl( \begin{smallmatrix} 1 & 2 & 3 & 4 & 5 & 6 & 7 & 8 & 9 & 10 & 11 \\ 2 & 6 & 4 & 7 & 8 & \fat{\red{9}} & \fat{\red{1}} & 5 & 10 & 11 & 3 \end{smallmatrix} \bigr) = (1, 2, 6, 9, 10, 11, 3 ,4, \fat{\red{7}})(\fat{\red{8}}, 5)$
	\bigskip
	
	\fbox{$\sigma_8 = \bigl( \begin{smallmatrix} 1 & 2 & 3 & 4 & 5 & 6 & 7 & 8 & 9 & 10 & 11 \\ 2 & 6 & 4 & 7 & 8 & 9 & \fat{\red{5}} & \fat{\red{1}} & 10 & 11 & 3 \end{smallmatrix} \bigr) = (1, 2, 6, \fat{\red{9}}, 10, 11, 3 ,4, 7, 5, \fat{\red{8}})$}
	\bigskip
	
	$\sigma_9 = \bigl( \begin{smallmatrix} 1 & 2 & 3 & 4 & 5 & 6 & 7 & 8 & 9 & 10 & 11 \\ 2 & 6 & 4 & 7 & 8 & 9 & 5 & \fat{\red{10}} & \fat{\red{1}} & 11 & 3 \end{smallmatrix} \bigr) = (1, 2, 6, \fat{\red{9}})(\fat{\red{10}}, 11, 3 ,4, 7, 5, 8)$
	\bigskip
	
	\fbox{$\sigma_{10} = \bigl( \begin{smallmatrix} 1 & 2 & 3 & 4 & 5 & 6 & 7 & 8 & 9 & 10 & 11 \\ 2 & 6 & 4 & 7 & 8 & 9 & 5 & 10 & \fat{\red{11}} & \fat{\red{1}} & 3 \end{smallmatrix} \bigr) = (1, 2, 6, 9, \fat{\red{11}}, 3, 4, 7, 5, 8, \fat{\red{10}})$}
	\bigskip
	
	$\sigma_{11} = \bigl( \begin{smallmatrix} 1 & 2 & 3 & 4 & 5 & 6 & 7 & 8 & 9 & 10 & 11 \\ 2 & 6 & 4 & 7 & 8 & 9 & 5 & 10 & 11 & \fat{\red{3}} & \fat{\red{1}} \end{smallmatrix} \bigr) = (1, 2, 6, 9, 11)( 3, 4, 7, 5, 8, 10)$
	
	\bigskip
\end{minipage}
\hspace{-2.5cm}\begin{minipage}[t]{0.2\textwidth}
	
	\vspace{1.15cm}
	
	{\em merge}
	\bigskip
	
	{\em merge}
	\bigskip
	
	{\em merge}
	\bigskip
	\vspace{0.3cm}
	
	{\em split}
	\bigskip
	
	{\em split}
	\bigskip
	
	{\em merge}
	\bigskip
	
	{\em merge}
	\bigskip
	\vspace{0.3cm}
	
	{\em split}
	\bigskip
	
	{\em merge}
	\bigskip
	\vspace{0.3cm}
	
	{\em split}
	
\end{minipage}
\hspace{-1.5cm}\begin{minipage}[t]{0.05\textwidth}
	
	\hspace{-1.2cm}{\bf no.\ cycles}
	\bigskip
	
	4
	\bigskip
	
	3
	\bigskip
	
	2
	\vspace{0.55cm}
	
	1
	\vspace{0.55cm}
	
	2
	\bigskip
	
	3
	\bigskip
	
	2
	\vspace{0.55cm}
	
	1
	\vspace{0.55cm}
	
	2
	\vspace{0.55cm}
	
	1
	\vspace{0.55cm}
	
	2
\end{minipage}

\newpage

\end{exa}


\begin{exa}\label{ex:sigmaPermutations3}
$$w = {\tt ccaaabcaac} \qquad \sigma = \bigl( \begin{smallmatrix} 1 & 2 & 3 & 4 & 5 & 6 & 7 & 8 & 9 & 10 \\ 7 & 8 & 1 & 2 & 3 & 6 & 9 & 4 & 5 & 10 \end{smallmatrix} \bigr) = (1,7,9,5,3)(2,8,4)(6)(10)$$

\centering
\begin{minipage}[t]{0.8\textwidth}
	
	\vspace{0.5cm}
	
	$\sigma_1 = \bigl( \begin{smallmatrix} \fat{\red{1}} & 2 & 3 & 4 & 5 & 6 & 7 & 8 & 9 & 10 & 11 \\ \fat{\red{1}} & 8 & 9 & 2 & 3 & 4 & 7 & 10 & 5 & 6 & 11 \end{smallmatrix} \bigr) = (\fat{\red{1}})(\fat{\red{2}},8,10,6,4)(3,9,5)(7)(11)$
	\bigskip
	
	$\sigma_2 = \bigl( \begin{smallmatrix} 1 & 2 & 3 & 4 & 5 & 6 & 7 & 8 & 9 & 10 & 11 \\ \fat{\red{8}} & \fat{\red{1}} & 9 & 2 & 3 & 4 & 7 & 10 & 5 & 6 & 11 \end{smallmatrix} \bigr) = (1,8,10,6,4,\fat{\red{2}})(\fat{\red{3}},9,5)(7)(11)$
	\bigskip
	
	$\sigma_3 = \bigl( \begin{smallmatrix} 1 & 2 & 3 & 4 & 5 & 6 & 7 & 8 & 9 & 10 & 11 \\ 8 & \fat{\red{9}} & \fat{\red{1}} & 2 & 3 & 4 & 7 & 10 & 5 & 6 & 11 \end{smallmatrix} \bigr) = (1,8,10,6,\fat{\red{4}},2,9,5,\fat{\red{3}})(7)(11)$
	\bigskip
	
	$\sigma_4 = \bigl( \begin{smallmatrix} 1 & 2 & 3 & 4 & 5 & 6 & 7 & 8 & 9 & 10 & 11 \\  8 & 9 & \fat{\red{2}} & \fat{\red{1}} & 3 & 4 & 7 & 10 & 5 & 6 & 11 \end{smallmatrix} \bigr) = (1,8,10,6,\fat{\red{4}})(2,9,\fat{\red{5}},3)(7)(11)$
	\bigskip
	
	$\sigma_5 = \bigl( \begin{smallmatrix} 1 & 2 & 3 & 4 & 5 & 6 & 7 & 8 & 9 & 10 & 11 \\ 8 & 9 & 2 & \fat{\red{3}} & \fat{\red{1}} & 4 & 7 & 10 & 5 & 6 & 11 \end{smallmatrix} \bigr) = (1,8,10,\fat{\red{6}},4,3,2,9,\fat{\red{5}})(7)(11)$
	\bigskip
	
	$\sigma_6 = \bigl( \begin{smallmatrix} 1 & 2 & 3 & 4 & 5 & 6 & 7 & 8 & 9 & 10 & 11 \\ 8 & 9 & 2 & 3 & \fat{\red{4}} & \fat{\red{1}} & 7 & 10 & 5 & 6 & 11 \end{smallmatrix} \bigr) = (1,8,10,\fat{\red{6}})(4,3,2,9,5)(\fat{\red{7}})(11)$
	\bigskip
	
	$\sigma_7 = \bigl( \begin{smallmatrix} 1 & 2 & 3 & 4 & 5 & 6 & 7 & 8 & 9 & 10 & 11 \\ 8 & 9 & 2 & 3 & 4 & \fat{\red{7}} & \fat{\red{1}} & 10 & 5 & 6 & 11 \end{smallmatrix} \bigr) = (1,\fat{\red{8}},10,6,\fat{\red{7}})(4,3,2,9,5)(11)$
	\bigskip
	
	$\sigma_8 = \bigl( \begin{smallmatrix} 1 & 2 & 3 & 4 & 5 & 6 & 7 & 8 & 9 & 10 & 11 \\ 8 & 9 & 2 & 3 & 4 & 7 & \fat{\red{10}} & \fat{\red{1}} & 5 & 6 & 11 \end{smallmatrix} \bigr) = (1,\fat{\red{8}})(10,6,7)(4,3,2,\fat{\red{9}},5)(11)$
	\bigskip
	
	$\sigma_9 = \bigl( \begin{smallmatrix} 1 & 2 & 3 & 4 & 5 & 6 & 7 & 8 & 9 & 10 & 11 \\  8 & 9 & 2 & 3 & 4 & 7 & 10 & \fat{\red{5}} & \fat{\red{1}} & 6 & 11 \end{smallmatrix} \bigr) = (1,8,5,4,3,2,\fat{\red{9}})(\fat{\red{10}},6,7)(11)$
	\bigskip
	
	$\sigma_{10} = \bigl( \begin{smallmatrix} 1 & 2 & 3 & 4 & 5 & 6 & 7 & 8 & 9 & 10 & 11 \\ 8 & 9 & 2 & 3 & 4 & 7 & 10 & 5 & \fat{\red{6}} & \fat{\red{1}} & 11 \end{smallmatrix} \bigr) = (1,8,5,4,3,2,9,6,7,\fat{\red{10}})(\fat{\red{11}})$
	\bigskip
	
	\fbox{$\sigma_{11} = \bigl( \begin{smallmatrix} 1 & 2 & 3 & 4 & 5 & 6 & 7 & 8 & 9 & 10 & 11 \\8 & 9 & 2 & 3 & 4 & 7 & 10 & 5 & 6 & \fat{\red{11}} & \fat{\red{1}} \end{smallmatrix} \bigr) = (1,8,5,4,3,2,9,6,7,10,11)$}
	
	\bigskip
\end{minipage}
\hspace{-2.5cm}\begin{minipage}[t]{0.2\textwidth}
	
	\vspace{1.15cm}
	
	{\em merge}
	\bigskip
	
	{\em merge}
	\bigskip
	
	{\em split}
	\bigskip
	
	{\em merge}
	\bigskip
	
	{\em split}
	\bigskip
	
	{\em merge}
	\bigskip
	
	{\em split}
	\bigskip
	
	{\em merge}
	\bigskip
	
	{\em merge}
	\bigskip
	
	{\em merge}
	
\end{minipage}
\hspace{-1.5cm}\begin{minipage}[t]{0.05\textwidth}
	
	\hspace{-1.2cm}{\bf no.\ cycles}
	\bigskip
	
	5
	\bigskip
	
	4
	\bigskip
	
	3
	\bigskip
	
	4
	\bigskip
	
	3
	\bigskip
	
	4
	\bigskip
	
	3
	\bigskip
	
	4
	\bigskip
	
	3
	\bigskip
	
	2
	\vspace{0.65cm}
	
	1
\end{minipage}

\end{exa}

\begin{figure}
	\captionsetup[subfigure]{font=footnotesize}
	\centering
	\caption{The splay tree implementation for Example~\ref{ex:sigmaPermutations1}.\label{fig:sigmaPermutations1}}
	\subfloat[$\sigma_1$]{%
		\scalebox{0.50}{%
			\begin{tikzpicture}[level distance=1.5cm,
			level 1/.style={sibling distance=1.5cm},
			level 2/.style={sibling distance=1cm}]
			
			\node[nodered] at (0,0) (root1){$1$};
			
			\node[nodered] at (1.3,0) {$2$};  
			
			\node[node, minimum size=.5cm] at (2.6,0){$7$}
			child{node[node, minimum size=.82cm]{$3$}}
			child{node[node, minimum size=.82cm]{$10$}};
			
			\node[node, minimum size=.82cm] at (5.2,0){$11$}
			child{node[node, minimum size=.82cm]{$8$}
				child{node[node, minimum size=.82cm]{$4$}}
				child[missing]{} }
			child{ node[node, minimum size=.82cm]{$6$} };
			
			\node[node, minimum size=.82cm] at (6.8,0) {$5$}
			child[missing]{}
			child{node[node, minimum size=.82cm]{$9$}}; 
			
			\end{tikzpicture}}%
	}
	\hspace{3em}
	\subfloat[$\sigma_2$]{%
		\scalebox{0.50}{%
			\begin{tikzpicture}[level distance=1.5cm,
			level 1/.style={sibling distance=1.5cm},
			level 2/.style={sibling distance=1cm}]
			
			\node[nodeone, minimum size=.82cm] at (0,0) (root1){$1$}
			child[missing]{}	
			child{node[nodered]{$2$}};
			
			\node[node] at (2.7,0){$7$}
			child{node[nodered]{$3$}}
			child{node[node, minimum size=.82cm]{$10$}};
			
			\node[node, minimum size=.82cm] at (5.8,0){$11$}
			child{node[node, minimum size=.82cm]{$8$}
				child{node[node, minimum size=.82cm]{$4$}}
				child[missing]{} }
			child{ node[node, minimum size=.82cm]{$6$} };
			
			\node[node, minimum size=.82cm] at (7.5,0) {$5$}
			child[missing]{}
			child{node[node, minimum size=.82cm]{$9$}};

			\end{tikzpicture}}%
	}
	
	\subfloat[$\sigma_3$]{%
		\scalebox{0.50}{%
			\begin{tikzpicture}[level distance=1.5cm,
			level 1/.style={sibling distance=1.5cm},
			level 2/.style={sibling distance=1.5cm},
			level 3/.style={sibling distance=1cm}]
			
			\node[node, minimum size=.82cm] at (0,0) (root1){$2$}
			child{node[nodeone, minimum size=.82cm]{$1$}}	
			child{node[node, minimum size=.82cm]{$10$}
				child{node[node, minimum size=.82cm]{$7$}}
				child{node[nodered]{3}}};
			
			\node[node, minimum size=.82cm] at (4.2,0){$11$}
			child{node[node, minimum size=.82cm]{$8$}
				child{node[nodered]{$4$}}
				child[missing]{} }
			child{ node[node, minimum size=.82cm]{$6$} };
			
			\node[node, minimum size=.82cm] at (5.7,0) {$5$}
			child[missing]{}
			child{node[node, minimum size=.82cm]{$9$}};

			\end{tikzpicture}}%
	}
	\hspace{3em}
	\subfloat[$\sigma_4$]{%
		\scalebox{0.50}{%
			\begin{tikzpicture}[level distance=1.5cm,
			level 1/.style={sibling distance=2.5cm},
			level 2/.style={sibling distance=1.8cm},
			level 3/.style={sibling distance=1.3cm}]
			
			\node[node, minimum size=.82cm] at (0,0) (root1){$3$}
			child{node[node, minimum size=.82cm]{10}
				child{node[node, minimum size=.82cm]{2}
					child{node[nodeone, minimum size=.82cm]{$1$}}
					child{node[node, minimum size=.82cm]{$7$}}}
				child[missing]{}}
			child{node[node, minimum size=.82cm]{$6$}
				child{node[node, minimum size=.82cm]{$11$}
					child{node[node, minimum size=.82cm]{$8$}}
					child[missing]{}}
				child{node[nodered]{$4$}}};
			
			\node[nodered] at (2,0) {$5$}
			child[missing]{}
			child{node[node, minimum size=.82cm]{$9$}};   
			
			\end{tikzpicture}}%
	}
	
	\subfloat[$\sigma_5$]{%
		\scalebox{0.50}{%
			\begin{tikzpicture}[level distance=1.5cm,
			level 1/.style={sibling distance=1.5cm},
			level 2/.style={sibling distance=1.5cm},
			level 3/.style={sibling distance=1.5cm}]
			
			\node[node, minimum size=.82cm] at (0,0) (root1){$4$}
			child{node[nodered]{6}
				child{node[node, minimum size=.82cm]{3}
					child{node[node, minimum size=.82cm]{10}
						child{node[node, minimum size=.82cm]{2}
							child{node[nodeone, minimum size=.82cm]{1}}
							child{node[node, minimum size=.82cm]{7}}}
						child[missing]{}}
					child{node[node, minimum size=.82cm]{11}
						child{node[node, minimum size=.82cm]{8}}
						child[missing]{}}}
				child[missing]}
			child{node[node, minimum size=.82cm]{9}
				child[missing]{}
				child{node[nodered]{5}}};
			\end{tikzpicture}}%
	}
	\hspace{3em}
	\subfloat[$\sigma_6$]{%
		\scalebox{0.50}{%
			\begin{tikzpicture}[level distance=1.5cm,
			level 1/.style={sibling distance=1.5cm},
			level 2/.style={sibling distance=1.5cm},
			level 3/.style={sibling distance=1.5cm}]
			
			\node[nodered] at (0,0) (root1){$6$}
			child{node[node, minimum size=.82cm]{3}
				child{node[node, minimum size=.82cm]{10}
					child{node[node, minimum size=.82cm]{2}
						child{node[nodeone, minimum size=.82cm]{1}}
						child{node[nodered]{7}}}
					child[missing]{}}
				child{node[node, minimum size=.82cm]{11}
					child{node[node, minimum size=.82cm]{8}}
					child[missing]{}}}
			child[missing]{};
			
			\node[node, minimum size=.82cm] at (2,0) {$4$}
			child[missing]{}
			child{node[node, minimum size=.82cm]{$9$}
				child[missing]{}
				child{node[node, minimum size=.82cm]{$5$}}};;
			
			\end{tikzpicture}}%
	}
	
	\subfloat[$\sigma_7$]{%
		\scalebox{0.50}{%
			\begin{tikzpicture}[level distance=1.5cm,
			level 1/.style={sibling distance=1.5cm},
			level 2/.style={sibling distance=1.5cm},
			level 3/.style={sibling distance=1.5cm}]

			\node[nodered] at (0,0) {$7$}
			child{node[node, minimum size=.82cm]{$2$}
				child{node[nodeone, minimum size=.82cm]{$1$}}
				child[missing]{}}
			child[missing]{};
			
			\node[node, minimum size=.82cm] at (2.5,0) {$3$}
			child{node[node, minimum size=.82cm]{10}}
			child{node[node, minimum size=.82cm]{6}
				child{node[node, minimum size=.82cm]{11}
					child{node[nodered]{8}}
					child[missing]{}}
				child[missing]{}};
			
			\node[node, minimum size=.82cm] at (4,0) {$4$}
			child[missing]{}
			child{node[node, minimum size=.82cm]{$9$}
				child[missing]{}
				child{node[node, minimum size=.82cm]{$5$}}};
			
			\end{tikzpicture}}%
	}
	\hspace{3em}
	\subfloat[$\sigma_8$]{%
		\scalebox{0.50}{%
			\begin{tikzpicture}[level distance=1.5cm,
			level 1/.style={sibling distance=1.5cm},
			level 2/.style={sibling distance=1.5cm},
			level 3/.style={sibling distance=1.5cm}]

			\node[node, minimum size=.82cm] at (0,0) {$7$}
			child{node[node, minimum size=.82cm]{$2$}
				child{node[nodeone, minimum size=.82cm]{$1$}}
				child[missing]{}}
			child{node[node, minimum size=.82cm]{$6$}
				child{node[node, minimum size=.82cm]{$11$}}
				child{node[nodered]{$8$}
					child{node[node, minimum size=.82cm]{$3$}
						child{node[node, minimum size=.82cm]{$10$}}
						child[missing]{}}
					child[missing]{}}};
			
			\node[node, minimum size=.82cm] at (4,0) {$4$}
			child[missing]{}
			child{node[nodered]{$9$}
				child[missing]{}
				child{node[node, minimum size=.82cm]{$5$}}};
			
			\end{tikzpicture}}%
	}
	
	\subfloat[$\sigma_9$]{%
		\scalebox{0.50}{%
			\begin{tikzpicture}[level distance=1.5cm,
			level 1/.style={sibling distance=2.5cm},
			level 2/.style={sibling distance=2.5cm},
			level 3/.style={sibling distance=1.4cm}]

			\node[node, minimum size=.82cm] at (0,0) {$8$}
			child{node[node, minimum size=.82cm]{$6$}
				child{node[node, minimum size=.82cm]{$7$}
					child{node[node, minimum size=.82cm]{$2$}
						child{node[nodeone, minimum size=.82cm]{$1$}}
						child[missing]{}}
					child{node[node, minimum size=.82cm]{$11$}}}
				child{node[node, minimum size=.82cm]{$3$}
					child{node[nodered]{$10$}}
					child[missing]{}}}
			child{node[node, minimum size=.82cm]{$5$}
				child[missing]{}
				child{node[nodered]{$9$}
					child{node[node, minimum size=.82cm]{$4$}}
					child[missing]}};

			\end{tikzpicture}}%
	}
	\hspace{3em}
	\subfloat[$\sigma_{10}$]{%
		\scalebox{0.50}{%
			\begin{tikzpicture}[level distance=1.5cm,
			level 1/.style={sibling distance=2cm},
			level 2/.style={sibling distance=1.8cm},
			level 3/.style={sibling distance=1.5cm}]

			\node[nodered] at (0,0) {$10$}
			child{node[node, minimum size=.82cm]{$6$}
				child{node[node, minimum size=.82cm]{$7$}
					child{node[node, minimum size=.82cm]{$2$}
						child{node[nodeone, minimum size=.82cm]{$1$}}
						child[missing]{}}
					child{node[nodered]{$11$}}}
				child[missing]{}}
			child[missing]{};
			
			\node[node, minimum size=.82cm] at (3,0) {$8$}
			child{node[node, minimum size=.82cm]{$3$}}
			child{node[node, minimum size=.82cm]{$5$}
				child[missing]{}
				child{node[node, minimum size=.82cm]{$9$}
					child{node[node, minimum size=.82cm]{$4$}}
					child[missing]{}}};
			
			\end{tikzpicture}}
	}
	
	\subfloat[$\sigma_{11}$]{%
		\scalebox{0.50}{%
			\begin{tikzpicture}[level distance=1.5cm,
			level 1/.style={sibling distance=2cm},
			level 2/.style={sibling distance=1.8cm},
			level 3/.style={sibling distance=1.5cm}]

			\node[node, minimum size=.82cm] at (0,0) {$11$}
			child{node[node, minimum size=.82cm]{$7$}
				child{node[node, minimum size=.82cm]{$2$}
					child{node[nodeone, minimum size=.82cm]{$1$}}
					child[missing]{}}
				child[missing]{}}
			child[missing]{};
			
			\node[node, minimum size=.82cm] at (2,0) {$10$}
			child{node[node, minimum size=.82cm]{$6$}}
			child[missing]{};
			
			\node[node, minimum size=.82cm] at (4,0) {$8$}
			child{node[node, minimum size=.82cm]{$3$}}
			child{node[node, minimum size=.82cm]{$5$}
				child[missing]{}
				child{node[node, minimum size=.82cm]{$9$}
					child{node[node, minimum size=.82cm]{$4$}}
					child[missing]{}}};
			
			\end{tikzpicture}}
	}
\end{figure}
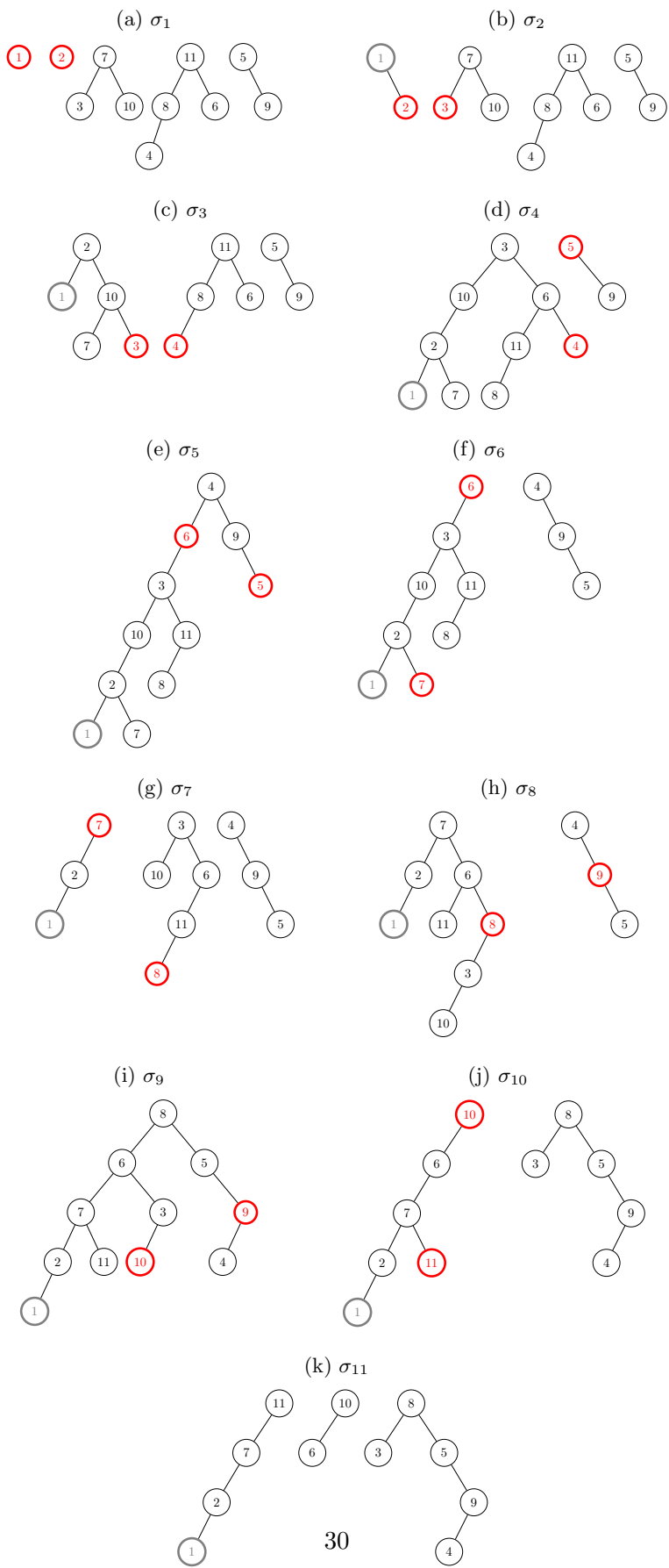


\newpage

\begin{center}
	

\begingroup
\setlength{\LTleft}{-20cm plus -1fill}
\setlength{\LTright}{\LTleft}

\bottomcaption{Statistics of words of length from 3 to 20 over a binary alphabet. Percentages rounded to nearest integer. See text for further details. 
	\label{tab:tablePercentages2}}
\tablehead{\hline ${\bf |\Sigma| = 2}$ && && && &&\multicolumn{2}{c|}{{\small BWTs of}} \\
	& $h(w)$ & {\small all} & \% & {\small noBWTs} & \% & {\small BWTs} & \% & {\small prim} & {\small pow} \\ \hline}

\begin{supertabular}{c || r | r | r | r | r | r | r || r | r |} \hline

$n=3$ & 0 & 1 & 12 & 1 & 25 & 0 & 0 & 0 & 0\\  
 & 1 & 6 & 75 & 3 & 75 & 3 & 75 & 1 & 2\\  
& 2 & 1 & 12 & 0 & 0 & 1 & 25 & 1 & 0\\  
  \hline
& {\small total}  & 8 & 100 & 4 & 100 & 4 & 100 & 2 & 2\\  
\hline \hline
$n=4$ & 0 & 3 & 19 & 3 & 30 & 0 & 0 & 0 & 0\\  
& 1 & 10 & 62 & 6 & 60 & 4 & 67 & 2 & 2\\  
& 2 & 3 & 19 & 1 & 10 & 2 & 33 & 1 & 1\\  
  \hline
& {\small total}  & 16 & 100 & 10 & 100 & 6 & 100 & 3 & 3\\  
\hline \hline
$n=5$ & 0 & 9 & 28 & 9 & 38 & 0 & 0 & 0 & 0\\  
& 1 & 16 & 50 & 11 & 46 & 5 & 62 & 3 & 2\\  
& 2 & 5 & 16 & 4 & 17 & 1 & 12 & 1 & 0\\  
& 3 & 2 & 6 & 0 & 0 & 2 & 25 & 2 & 0\\  
  \hline
& {\small total}  & 32 & 100 & 24 & 100 & 8 & 100 & 6 & 2\\  
\hline \hline
$n=6$ & 0 & 21 & 33 & 21 & 42 & 0 & 0 & 0 & 0\\  
& 1 & 28 & 44 & 20 & 40 & 8 & 57 & 6 & 2\\  
& 2 & 9 & 14 & 6 & 12 & 3 & 21 & 1 & 2\\  
& 3 & 6 & 9 & 3 & 6 & 3 & 21 & 2 & 1\\  
  \hline
& {\small total}  & 64 & 100 & 50 & 100 & 14 & 100 & 9 & 5\\  
\hline \hline
$n=7$ & 0 & 51 & 40 & 51 & 47 & 0 & 0 & 0 & 0\\  
& 1 & 46 & 36 & 35 & 32 & 11 & 55 & 9 & 2\\  
& 2 & 14 & 11 & 13 & 12 & 1 & 5 & 1 & 0\\  
& 3 & 14 & 11 & 9 & 8 & 5 & 25 & 5 & 0\\  
& 4 & 3 & 2 & 0 & 0 & 3 & 15 & 3 & 0\\  
  \hline
& {\small total}  & 128 & 100 & 108 & 100 & 20 & 100 & 18 & 2\\  
\hline \hline
$n=8$ & 0 & 111 & 43 & 111 & 50 & 0 & 0 & 0 & 0\\  
& 1 & 84 & 33 & 64 & 29 & 20 & 56 & 18 & 2\\  
& 2 & 20 & 8 & 19 & 9 & 1 & 3 & 1 & 0\\  
& 3 & 32 & 12 & 21 & 10 & 11 & 31 & 8 & 3\\  
& 4 & 9 & 4 & 5 & 2 & 4 & 11 & 3 & 1\\  
  \hline
& {\small total}  & 256 & 100 & 220 & 100 & 36 & 100 & 30 & 6\\  
\hline \hline
$n=9$ & 0 & 243 & 47 & 243 & 54 & 0 & 0 & 0 & 0\\  
& 1 & 150 & 29 & 118 & 26 & 32 & 53 & 30 & 2\\  
& 2 & 31 & 6 & 29 & 6 & 2 & 3 & 2 & 0\\  
& 3 & 56 & 11 & 48 & 11 & 8 & 13 & 7 & 1\\  
& 4 & 28 & 5 & 14 & 3 & 14 & 23 & 13 & 1\\  
& 5 & 4 & 1 & 0 & 0 & 4 & 7 & 4 & 0\\  
  \hline
& {\small total}  & 512 & 100 & 452 & 100 & 60 & 100 & 56 & 4\\  
\hline \hline
$n=10$ & 0 & 515 & 50 & 515 & 56 & 0 & 0 & 0 & 0\\  
& 1 & 274 & 27 & 216 & 24 & 58 & 54 & 56 & 2\\  
& 2 & 53 & 5 & 48 & 5 & 5 & 5 & 5 & 0\\  
& 3 & 98 & 10 & 81 & 9 & 17 & 16 & 14 & 3\\  
& 4 & 70 & 7 & 48 & 5 & 22 & 20 & 19 & 3\\  
& 5 & 14 & 1 & 8 & 1 & 6 & 6 & 5 & 1\\  
  \hline
& {\small total}  & 1024 & 100 & 916 & 100 & 108 & 100 & 99 & 9\\  
\hline \hline
$n=11$ & 0 & 1088 & 53 & 1088 & 58 & 0 & 0 & 0 & 0\\  
& 1 & 494 & 24 & 393 & 21 & 101 & 54 & 99 & 2\\  
& 2 & 104 & 5 & 96 & 5 & 8 & 4 & 8 & 0\\  
& 3 & 164 & 8 & 147 & 8 & 17 & 9 & 17 & 0\\  
& 4 & 142 & 7 & 114 & 6 & 28 & 15 & 28 & 0\\  
& 5 & 50 & 2 & 22 & 1 & 28 & 15 & 28 & 0\\  
& 6 & 6 & $<0.5$ & 0 & 0 & 6 & 3 & 6 & 0\\  
  \hline
& {\small total}  & 2048 & 100 & 1860 & 100 & 188 & 100 & 186 & 2\\  
\hline \hline
$n=12$ & 0 & 2258 & 55 & 2258 & 60 & 0 & 0 & 0 & 0\\  
& 1 & 914 & 22 & 724 & 19 & 190 & 54 & 188 & 2\\  
& 2 & 200 & 5 & 183 & 5 & 17 & 5 & 17 & 0\\  
& 3 & 288 & 7 & 252 & 7 & 36 & 10 & 34 & 2\\  
& 4 & 284 & 7 & 224 & 6 & 60 & 17 & 51 & 9\\  
& 5 & 130 & 3 & 90 & 2 & 40 & 11 & 37 & 3\\  
& 6 & 22 & 1 & 13 & $<0.5$ & 9 & 3 & 8 & 1\\  
  \hline
& {\small total}  & 4096 & 100 & 3744 & 100 & 352 & 100 & 335 & 17\\  
\hline \hline
$n=13$ & 0 & 4679 & 57 & 4679 & 62 & 0 & 0 & 0 & 0\\  
& 1 & 1680 & 21 & 1339 & 18 & 341 & 54 & 339 & 2\\  
& 2 & 388 & 5 & 365 & 5 & 23 & 4 & 23 & 0\\  
& 3 & 536 & 7 & 478 & 6 & 58 & 9 & 58 & 0\\  
& 4 & 522 & 6 & 455 & 6 & 67 & 11 & 67 & 0\\  
& 5 & 292 & 4 & 209 & 3 & 83 & 13 & 83 & 0\\  
& 6 & 85 & 1 & 35 & $<0.5$ & 50 & 8 & 50 & 0\\  
& 7 & 10 & $<0.5$ & 0 & 0 & 10 & 2 & 10 & 0\\  
  \hline
& {\small total}  & 8192 & 100 & 7560 & 100 & 632 & 100 & 630 & 2\\  
\hline \hline
$n=14$ & 0 & 9601 & 59 & 9601 & 63 & 0 & 0 & 0 & 0\\  
& 1 & 3140 & 19 & 2498 & 16 & 642 & 54 & 640 & 2\\  
& 2 & 748 & 5 & 696 & 5 & 52 & 4 & 52 & 0\\  
& 3 & 1024 & 6 & 907 & 6 & 117 & 10 & 114 & 3\\  
& 4 & 982 & 6 & 843 & 6 & 139 & 12 & 134 & 5\\  
& 5 & 620 & 4 & 475 & 3 & 145 & 12 & 139 & 6\\  
& 6 & 235 & 1 & 161 & 1 & 74 & 6 & 70 & 4\\  
& 7 & 34 & $<0.5$ & 21 & $<0.5$ & 13 & 1 & 12 & 1\\  
  \hline
& {\small total}  & 16384 & 100 & 15202 & 100 & 1182 & 100 & 1161 & 21\\  
\hline \hline
$n=15$ & 0 & 19664 & 60 & 19664 & 64 & 0 & 0 & 0 & 0\\  
& 1 & 5874 & 18 & 4695 & 15 & 1179 & 54 & 1177 & 2\\  
& 2 & 1464 & 4 & 1381 & 5 & 83 & 4 & 83 & 0\\  
& 3 & 1940 & 6 & 1775 & 6 & 165 & 8 & 165 & 0\\  
& 4 & 1880 & 6 & 1637 & 5 & 243 & 11 & 242 & 1\\  
& 5 & 1218 & 4 & 991 & 3 & 227 & 10 & 222 & 5\\  
& 6 & 574 & 2 & 380 & 1 & 194 & 9 & 192 & 2\\  
& 7 & 140 & $<0.5$ & 53 & $<0.5$ & 87 & 4 & 87 & 0\\  
& 8 & 14 & $<0.5$ & 0 & 0 & 14 & 1 & 14 & 0\\  
  \hline
& {\small total}  & 32768 & 100 & 30576 & 100 & 2192 & 100 & 2182 & 10\\  
\hline \hline
$n=16$ & 0 & 40094 & 61 & 40094 & 65 & 0 & 0 & 0 & 0\\  
& 1 & 11062 & 17 & 8843 & 14 & 2219 & 54 & 2217 & 2\\  
& 2 & 2882 & 4 & 2709 & 4 & 173 & 4 & 173 & 0\\  
& 3 & 3702 & 6 & 3346 & 5 & 356 & 9 & 353 & 3\\  
& 4 & 3622 & 6 & 3157 & 5 & 465 & 11 & 459 & 6\\  
& 5 & 2426 & 4 & 1988 & 3 & 438 & 11 & 427 & 11\\  
& 6 & 1298 & 2 & 980 & 2 & 318 & 8 & 309 & 9\\  
& 7 & 402 & 1 & 273 & $<0.5$ & 129 & 3 & 125 & 4\\  
& 8 & 48 & $<0.5$ & 30 & $<0.5$ & 18 & $<0.5$ & 17 & 1\\  
  \hline
& {\small total}  & 65536 & 100 & 61420 & 100 & 4116 & 100 & 4080 & 36\\  
\hline \hline
$n=17$ & 0 & 81602 & 62 & 81602 & 66 & 0 & 0 & 0 & 0\\  
& 1 & 20898 & 16 & 16758 & 14 & 4140 & 54 & 4138 & 2\\  
& 2 & 5711 & 4 & 5423 & 4 & 288 & 4 & 288 & 0\\  
& 3 & 7146 & 5 & 6589 & 5 & 557 & 7 & 557 & 0\\  
& 4 & 6863 & 5 & 6139 & 5 & 724 & 9 & 724 & 0\\  
& 5 & 4820 & 4 & 4020 & 3 & 800 & 10 & 800 & 0\\  
& 6 & 2736 & 2 & 2112 & 2 & 624 & 8 & 624 & 0\\  
& 7 & 1042 & 1 & 639 & 1 & 403 & 5 & 403 & 0\\  
& 8 & 234 & $<0.5$ & 78 & $<0.5$ & 156 & 2 & 156 & 0\\  
& 9 & 20 & $<0.5$ & 0 & 0 & 20 & $<0.5$ & 20 & 0\\  
  \hline
& {\small total}  & 131072 & 100 & 123360 & 100 & 7712 & 100 & 7710 & 2\\  
\hline \hline
$n=18$ & 0 & 165632 & 63 & 165632 & 67 & 0 & 0 & 0 & 0\\  
& 1 & 39704 & 15 & 31871 & 13 & 7833 & 54 & 7831 & 2\\  
& 2 & 11281 & 4 & 10696 & 4 & 585 & 4 & 585 & 0\\  
& 3 & 13798 & 5 & 12641 & 5 & 1157 & 8 & 1153 & 4\\  
& 4 & 13167 & 5 & 11672 & 5 & 1495 & 10 & 1485 & 10\\  
& 5 & 9620 & 4 & 8113 & 3 & 1507 & 10 & 1492 & 15\\  
& 6 & 5722 & 2 & 4579 & 2 & 1143 & 8 & 1119 & 24\\  
& 7 & 2450 & 1 & 1818 & 1 & 632 & 4 & 622 & 10\\  
& 8 & 696 & $<0.5$ & 474 & $<0.5$ & 222 & 2 & 218 & 4\\  
& 9 & 74 & $<0.5$ & 46 & $<0.5$ & 28 & $<0.5$ & 27 & 1\\  
  \hline
& {\small total}  & 262144 & 100 & 247542 & 100 & 14602 & 100 & 14532 & 70\\  
\hline \hline
$n=19$ & 0 & 335796 & 64 & 335796 & 68 & 0 & 0 & 0 & 0\\  
& 1 & 75616 & 14 & 60878 & 12 & 14738 & 53 & 14736 & 2\\  
& 2 & 22341 & 4 & 21347 & 4 & 994 & 4 & 994 & 0\\  
& 3 & 26820 & 5 & 24874 & 5 & 1946 & 7 & 1946 & 0\\  
& 4 & 25222 & 5 & 22776 & 5 & 2446 & 9 & 2446 & 0\\  
& 5 & 18916 & 4 & 16298 & 3 & 2618 & 9 & 2618 & 0\\  
& 6 & 11791 & 2 & 9520 & 2 & 2271 & 8 & 2271 & 0\\  
& 7 & 5428 & 1 & 3950 & 1 & 1478 & 5 & 1478 & 0\\  
& 8 & 1932 & $<0.5$ & 1131 & $<0.5$ & 801 & 3 & 801 & 0\\  
& 9 & 396 & $<0.5$ & 122 & $<0.5$ & 274 & 1 & 274 & 0\\  
& 10 & 30 & $<0.5$ & 0 & 0 & 30 & $<0.5$ & 30 & 0\\  
  \hline
& {\small total}  & 524288 & 100 & 496692 & 100 & 27596 & 100 & 27594 & 2\\  
\hline \hline
$n=20$ & 0 & 679618 & 65 & 679618 & 68 & 0 & 0 & 0 & 0\\  
& 1 & 144564 & 14 & 116554 & 12 & 28010 & 53 & 28008 & 2\\  
& 2 & 44031 & 4 & 42048 & 4 & 1983 & 4 & 1983 & 0\\  
& 3 & 52254 & 5 & 48274 & 5 & 3980 & 8 & 3971 & 9\\  
& 4 & 48598 & 5 & 43639 & 4 & 4959 & 9 & 4950 & 9\\  
& 5 & 37512 & 4 & 32411 & 3 & 5101 & 10 & 5084 & 17\\  
& 6 & 24195 & 2 & 20011 & 2 & 4184 & 8 & 4159 & 25\\  
& 7 & 11780 & 1 & 9185 & 1 & 2595 & 5 & 2564 & 31\\  
& 8 & 4718 & $<0.5$ & 3471 & $<0.5$ & 1247 & 2 & 1234 & 13\\  
& 9 & 1198 & $<0.5$ & 809 & $<0.5$ & 389 & 1 & 385 & 4\\  
& 10 & 108 & $<0.5$ & 68 & $<0.5$ & 40 & $<0.5$ & 39 & 1\\  
  \hline
& {\small total}  & 1048576 & 100 & 996088 & 100 & 52488 & 100 & 52377 & 111\\ 
\hline \hline

\end{supertabular}
\endgroup

	
\vspace{0.5cm}
\begingroup
\setlength{\LTleft}{-20cm plus -1fill}
\setlength{\LTright}{\LTleft}

\bottomcaption{Statistics of words of length 3 to 20 over an alphabet of size 3. Percentages rounded to nearest integer. See text for further details. 
	\label{tab:tablePercentages3}}
\tablehead{\hline ${\bf |\Sigma| = 3}$ && && && &&\multicolumn{2}{c|}{{\small BWTs of}} \\
	& $h(w)$ & {\small all} & \% & {\small noBWTs} & \% & {\small BWTs} & \% & {\small prim} & {\small pow} \\ \hline}

\begin{supertabular}{c || r | r | r | r | r | r | r || r | r |} \hline
$n=3$ & 0 & 4 & 15 & 4 & 25 & 0 & 0 & 0 & 0\\  
& 1 & 19 & 70 & 12 & 75 & 7 & 64 & 4 & 3\\  
& 2 & 4 & 15 & 0 & 0 & 4 & 36 & 4 & 0\\  
 \hline
& {\small total} & 27 & 100 & 16 & 100 & 11 & 100 & 8 & 3\\  
\hline \hline
$n=4$ & 0 & 18 & 22 & 18 & 32 & 0 & 0 & 0 & 0\\  
& 1 & 45 & 56 & 31 & 54 & 14 & 58 & 11 & 3\\  
& 2 & 18 & 22 & 8 & 14 & 10 & 42 & 7 & 3\\  
 \hline
& {\small total} & 81 & 100 & 57 & 100 & 24 & 100 & 18 & 6\\  
\hline \hline
$n=5$ & 0 & 74 & 30 & 74 & 39 & 0 & 0 & 0 & 0\\  
& 1 & 109 & 45 & 83 & 43 & 26 & 51 & 23 & 3\\  
& 2 & 46 & 19 & 35 & 18 & 11 & 22 & 11 & 0\\  
& 3 & 14 & 6 & 0 & 0 & 14 & 27 & 14 & 0\\  
 \hline
& {\small total} & 243 & 100 & 192 & 100 & 51 & 100 & 48 & 3\\  
\hline \hline
$n=6$ & 0 & 258 & 35 & 258 & 43 & 0 & 0 & 0 & 0\\  
& 1 & 277 & 38 & 212 & 35 & 65 & 50 & 62 & 3\\  
& 2 & 130 & 18 & 94 & 16 & 36 & 28 & 29 & 7\\  
& 3 & 64 & 9 & 35 & 6 & 29 & 22 & 25 & 4\\  
 \hline
& {\small total} & 729 & 100 & 599 & 100 & 130 & 100 & 116 & 14\\  
\hline \hline
$n=7$ & 0 & 884 & 40 & 884 & 47 & 0 & 0 & 0 & 0\\  
& 1 & 709 & 32 & 564 & 30 & 145 & 46 & 142 & 3\\  
& 2 & 348 & 16 & 290 & 15 & 58 & 18 & 58 & 0\\  
& 3 & 202 & 9 & 134 & 7 & 68 & 22 & 68 & 0\\  
& 4 & 44 & 2 & 0 & 0 & 44 & 14 & 44 & 0\\  
 \hline
& {\small total} & 2187 & 100 & 1872 & 100 & 315 & 100 & 312 & 3\\  
\hline \hline
$n=8$ & 0 & 2870 & 44 & 2870 & 50 & 0 & 0 & 0 & 0\\  
& 1 & 1897 & 29 & 1509 & 26 & 388 & 47 & 385 & 3\\  
& 2 & 932 & 14 & 766 & 13 & 166 & 20 & 164 & 2\\  
& 3 & 648 & 10 & 458 & 8 & 190 & 23 & 176 & 14\\  
& 4 & 214 & 3 & 124 & 2 & 90 & 11 & 85 & 5\\  
 \hline
& {\small total} & 6561 & 100 & 5727 & 100 & 834 & 100 & 810 & 24\\  
\hline \hline
$n=9$ & 0 & 9208 & 47 & 9208 & 53 & 0 & 0 & 0 & 0\\  
& 1 & 5135 & 26 & 4164 & 24 & 971 & 44 & 968 & 3\\  
& 2 & 2558 & 13 & 2193 & 13 & 365 & 17 & 365 & 0\\  
& 3 & 1834 & 9 & 1452 & 8 & 382 & 17 & 378 & 4\\  
& 4 & 810 & 4 & 471 & 3 & 339 & 15 & 335 & 4\\  
& 5 & 138 & 1 & 0 & 0 & 138 & 6 & 138 & 0\\  
 \hline
& {\small total} & 19683 & 100 & 17488 & 100 & 2195 & 100 & 2184 & 11\\  
\hline \hline
$n=10$ & 0 & 29024 & 49 & 29024 & 55 & 0 & 0 & 0 & 0\\  
& 1 & 14059 & 24 & 11438 & 22 & 2621 & 44 & 2618 & 3\\  
& 2 & 7148 & 12 & 6085 & 11 & 1063 & 18 & 1059 & 4\\  
& 3 & 5290 & 9 & 4184 & 8 & 1106 & 19 & 1088 & 18\\  
& 4 & 2816 & 5 & 1965 & 4 & 851 & 14 & 828 & 23\\  
& 5 & 712 & 1 & 419 & 1 & 293 & 5 & 287 & 6\\  
 \hline
& {\small total} & 59049 & 100 & 53115 & 100 & 5934 & 100 & 5880 & 54\\  
\hline \hline
$n=11$ & 0 & 90788 & 51 & 90788 & 56 & 0 & 0 & 0 & 0\\  
& 1 & 38917 & 22 & 32007 & 20 & 6910 & 43 & 6907 & 3\\  
& 2 & 20210 & 11 & 17722 & 11 & 2488 & 15 & 2488 & 0\\  
& 3 & 15124 & 9 & 12496 & 8 & 2628 & 16 & 2628 & 0\\  
& 4 & 8554 & 5 & 6415 & 4 & 2139 & 13 & 2139 & 0\\  
& 5 & 3102 & 2 & 1612 & 1 & 1490 & 9 & 1490 & 0\\  
& 6 & 452 & 0 & 0 & 0 & 452 & 3 & 452 & 0\\  
 \hline
& {\small total} & 177147 & 100 & 161040 & 100 & 16107 & 100 & 16104 & 3\\  
\hline \hline
$n=12$ & 0 & 281576 & 53 & 281576 & 58 & 0 & 0 & 0 & 0\\  
& 1 & 108657 & 20 & 89704 & 18 & 18953 & 43 & 18950 & 3\\  
& 2 & 57142 & 11 & 50035 & 10 & 7107 & 16 & 7098 & 9\\  
& 3 & 44168 & 8 & 36405 & 7 & 7763 & 18 & 7736 & 27\\  
& 4 & 25912 & 5 & 19992 & 4 & 5920 & 13 & 5855 & 65\\  
& 5 & 11568 & 2 & 7933 & 2 & 3635 & 8 & 3598 & 37\\  
& 6 & 2418 & 0 & 1428 & 0 & 990 & 2 & 983 & 7\\  
 \hline
& {\small total} & 531441 & 100 & 487073 & 100 & 44368 & 100 & 44220 & 148\\  
\hline \hline
$n=13$ & 0 & 869386 & 55 & 869386 & 59 & 0 & 0 & 0 & 0\\  
& 1 & 305877 & 19 & 254558 & 17 & 51319 & 42 & 51316 & 3\\  
& 2 & 162798 & 10 & 145275 & 10 & 17523 & 14 & 17523 & 0\\  
& 3 & 128166 & 8 & 108909 & 7 & 19257 & 16 & 19257 & 0\\  
& 4 & 77028 & 5 & 61255 & 4 & 15773 & 13 & 15773 & 0\\  
& 5 & 37740 & 2 & 26732 & 2 & 11008 & 9 & 11008 & 0\\  
& 6 & 11756 & 1 & 5565 & 0 & 6191 & 5 & 6191 & 0\\  
& 7 & 1572 & 0 & 0 & 0 & 1572 & 1 & 1572 & 0\\  
 \hline
& {\small total} & 1594323 & 100 & 1471680 & 100 & 122643 & 100 & 122640 & 3\\  
\hline \hline
$n=14$ & 0 & 2671854 & 56 & 2671854 & 60 & 0 & 0 & 0 & 0\\  
& 1 & 865371 & 18 & 722987 & 16 & 142384 & 42 & 142381 & 3\\  
& 2 & 465112 & 10 & 415044 & 9 & 50068 & 15 & 50047 & 21\\  
& 3 & 373298 & 8 & 317142 & 7 & 56156 & 16 & 56095 & 61\\  
& 4 & 232528 & 5 & 187196 & 4 & 45332 & 13 & 45255 & 77\\  
& 5 & 119972 & 3 & 90685 & 2 & 29287 & 9 & 29196 & 91\\  
& 6 & 46330 & 1 & 31229 & 1 & 15101 & 4 & 15044 & 57\\  
& 7 & 8504 & 0 & 5030 & 0 & 3474 & 1 & 3466 & 8\\  
 \hline
& {\small total} & 4782969 & 100 & 4441167 & 100 & 341802 & 100 & 341484 & 318\\  
\hline \hline
$n=15$ & 0 & 8187156 & 57 & 8187156 & 61 & 0 & 0 & 0 & 0\\  
& 1 & 2463913 & 17 & 2071284 & 15 & 392629 & 41 & 392626 & 3\\  
& 2 & 1336872 & 9 & 1209020 & 9 & 127852 & 13 & 127852 & 0\\  
& 3 & 1084672 & 8 & 941090 & 7 & 143582 & 15 & 143581 & 1\\  
& 4 & 692916 & 5 & 571012 & 4 & 121904 & 13 & 121892 & 12\\  
& 5 & 370878 & 3 & 284657 & 2 & 86221 & 9 & 86193 & 28\\  
& 6 & 161896 & 1 & 108262 & 1 & 53634 & 6 & 53619 & 15\\  
& 7 & 45028 & 0 & 19791 & 0 & 25237 & 3 & 25237 & 0\\  
& 8 & 5576 & 0 & 0 & 0 & 5576 & 1 & 5576 & 0\\  
 \hline
& {\small total} & 14348907 & 100 & 13392272 & 100 & 956635 & 100 & 956576 & 59\\  
\hline \hline
$n=16$ & 0 & 25017316 & 58 & 25017316 & 62 & 0 & 0 & 0 & 0\\  
& 1 & 7038391 & 16 & 5938764 & 15 & 1099627 & 41 & 1099624 & 3\\  
& 2 & 3848862 & 9 & 3483455 & 9 & 365407 & 14 & 365361 & 46\\  
& 3 & 3163798 & 7 & 2745285 & 7 & 418513 & 16 & 418378 & 135\\  
& 4 & 2073906 & 5 & 1721549 & 4 & 352357 & 13 & 352192 & 165\\  
& 5 & 1148454 & 3 & 908307 & 2 & 240147 & 9 & 239934 & 213\\  
& 6 & 541080 & 1 & 400379 & 1 & 140701 & 5 & 140515 & 186\\  
& 7 & 184474 & 0 & 122823 & 0 & 61651 & 2 & 61574 & 77\\  
& 8 & 30440 & 0 & 17999 & 0 & 12441 & 0 & 12432 & 9\\  
 \hline
& {\small total} & 43046721 & 100 & 40355877 & 100 & 2690844 & 100 & 2690010 & 834\\  
\hline \hline
$n=17$ & 0 & 76291362 & 59 & 76291362 & 63 & 0 & 0.00 & 0 & 0\\  
& 1 & 20204417 & 16 & 17133621 & 14 & 3070796 & 40 & 3070793 & 3\\  
& 2 & 11133250 & 9 & 10176801 & 8 & 956449 & 13 & 956449 & 0\\  
& 3 & 9223752 & 7 & 8128132 & 7 & 1095620 & 14 & 1095620 & 0\\  
& 4 & 6158862 & 5 & 5205889 & 4 & 952973 & 13 & 952973 & 0\\  
& 5 & 3503564 & 3 & 2809002 & 2 & 694562 & 9 & 694562 & 0\\  
& 6 & 1744222 & 1 & 1291991 & 1 & 452231 & 6 & 452231 & 0\\  
& 7 & 686708 & 1 & 435324 & $<0.5$ & 251384 & 3 & 251384 & 0\\  
& 8 & 173800 & $<0.5$ & 71558 & $<0.5$ & 102242 & 1 & 102242 & 0\\  
& 9 & 20226 & $<0.5$ & 0 & 0 & 20226 & $<0.5$ & 20226 & 0\\  
 \hline
& {\small total} & 129140163 & 100 & 121543680 & 100 & 7596483 & 100 & 7596480 & 3\\  
\hline \hline
$n=18$ & 0 & 232225836 & 60 & 232225836 & 63 & 0 & 0.00 & 0 & 0\\  
& 1 & 58136821 & 15 & 49468989 & 14 & 8667832 & 40 & 8667829 & 3\\  
& 2 & 32239908 & 8 & 29501333 & 8 & 2738575 & 13 & 2738462 & 113\\  
& 3 & 26967788 & 7 & 23775236 & 7 & 3192552 & 15 & 3192227 & 325\\  
& 4 & 18345718 & 5 & 15576819 & 4 & 2768899 & 13 & 2768490 & 409\\  
& 5 & 10706282 & 3 & 8733762 & 2 & 1972520 & 9 & 1972050 & 470\\  
& 6 & 5557276 & 1 & 4323149 & 1 & 1234127 & 6 & 1233560 & 567\\  
& 7 & 2396038 & 1 & 1741558 & $<0.5$ & 654480 & 3 & 654160 & 320\\  
& 8 & 733114 & $<0.5$ & 483308 & $<0.5$ & 249806 & 1 & 249709 & 97\\  
& 9 & 111708 & $<0.5$ & 65957 & $<0.5$ & 45751 & $<0.5$ & 45741 & 10\\  
 \hline
& {\small total} & 387420489 & 100 & 365895947 & 100 & 21524542 & 100 & 21522228 & 2314\\  
\hline \hline
$n=19$ & 0 & 705882210 & 61 & 705882210 & 64 & 0 & 0 & 0 & 0\\  
& 1 & 167901165 & 14 & 143474476 & 13 & 24426689 & 40 & 24426686 & 3\\  
& 2 & 93739960 & 8 & 86429556 & 8 & 7310404 & 12 & 7310404 & 0\\  
& 3 & 78834558 & 7 & 70340339 & 6 & 8494219 & 13 & 8494219 & 0\\  
& 4 & 54399352 & 5 & 46833919 & 4 & 7565433 & 12 & 7565433 & 0\\  
& 5 & 32418938 & 3 & 26772240 & 2 & 5646698 & 9 & 5646698 & 0\\  
& 6 & 17384924 & 2 & 13585965 & 1 & 3798959 & 6 & 3798959 & 0\\  
& 7 & 8057214 & 1 & 5762779 & 1 & 2294435 & 4 & 2294435 & 0\\  
& 8 & 2889010 & $<0.5$ & 1743235 & $<0.5$ & 1145775 & 2 & 1145775 & 0\\  
& 9 & 678872 & $<0.5$ & 265089 & $<0.5$ & 413783 & 1 & 413783 & 0\\  
& 10 & 75264 & $<0.5$ & 0 & 0 & 75264 & $<0.5$ & 75264 & 0\\  
 \hline
& {\small total} & 1162261467 & 100 & 1101089808 & 100 & 61171659 & 100 & 61171656 & 3\\
	\hline
	\hline	
$n=20$ & 0 & 2142855232 & 61 & 2142855232 & 65 & 0 & 0 & 0 & 0\\	
&1 & 485799027 & 14 & 416406087 & 13 & 69392940 & 40 & 69392937 & 3	
\\
&2 & 272773366 & 8 & 251795904 & 8 & 20977462 & 12 & 20977186 & 276 \\
&3 & 230944816 & 7 & 206174661 & 6 & 24770155 & 14 & 24769338 & 817 \\
&4 & 161668972 & 5 & 139645685 & 4 & 22023287 & 13 & 22022351 & 936 \\
&5 & 98349528 & 3 & 82124392 & 3 & 16225136 & 9 & 16224109 & 1027 \\
&6 & 54140632 & 2 & 43539360 & 1 & 10601272 & 6 & 10600113 & 1159 \\
&7 & 26452876 & 1 & 20257612 & 1 & 6195264 & 4 & 6194176 & 1088 \\
&8 & 10452426 & $<0.5$ & 7480470 & $<0.5$ & 2971956 & 2 & 2971438 & 518 \\
&9 & 2927926 & $<0.5$ & 1915361 & $<0.5$ & 1012565 & 1 & 1012448 & 117 \\
&10 & 419600 & $<0.5$ & 247421 & $<0.5$ & 172179 & $<0.5$ & 172168 & 11 \\
 \hline
& {\small total} & 3486784401 & 100 & 3312442185 & 100 & 174342216 & 100 & 174336264 & 5952\\
	\hline
\hline
\end{supertabular}
\endgroup
	
\end{center}

\end{document}